\documentclass[11pt]{article}
\usepackage[cmex10]{amsmath}
\usepackage{amsthm}
\usepackage{amssymb,amsfonts,textcomp}
\usepackage{latexsym}
\usepackage[nospace, compress]{cite}
\usepackage[noend, ruled, noline]{algorithm2e}

\usepackage{graphicx} 
\usepackage{latexsym}
\usepackage{multirow}
\usepackage{rotating}
\usepackage{url}
\usepackage{color}
\usepackage{paralist}

\newcommand{\shortcite}[1]{{{\cite{#1}}}}

\newcommand{\fixlist}{\addtolength{\itemsep}{0pt}}

\newcommand{\w}{{{{\mathrm{W}}}}}

\newcommand{\util}{{{\mathrm{ut}}}}
\newcommand{\Top}{{{\mathrm{Top}}}}

\DeclareMathOperator*{\argmax}{{{{\mathrm{argmax}}}}}

\def\med{\mathrm{med}^{(K)}}
\def\best{\mathrm{best}^{(K)}}
\def\ap{\mathrm{aprog[a]}^{(K)}}
\def\gp{\mathrm{gprog}[p]^{(K)}}

\newtheorem{theorem}{Theorem}
\newtheorem{definition}[theorem]{Definition}
\newtheorem{lemma}[theorem]{Lemma}
\newtheorem{observation}{Observation}
\newtheorem{proposition}[theorem]{Proposition}
\newtheorem{corollary}[theorem]{Corollary}

\theoremstyle{definition}
\newtheorem{example}{Example}

\newenvironment{proof-sketch}{\noindent\textit{Proof sketch.}\quad}{\hfill\medskip}

\newcommand{\np}{{\mathrm{NP}}}

\newcommand{\p}{{\mathrm{P}}}

\newcommand{\naturals}{{{\mathbb{N}}}}
\newcommand{\realsplus}{{{\mathbb{R}}}_{+}}

\newcommand{\calA}{{{\mathcal{A}}}}

\newcommand{\calS}{{{\mathcal{S}}}}
\newcommand{\calC}{{{\mathcal{C}}}}
\newcommand{\owa}{{{{\alpha}}}}
\newcommand{\OWA}{{{{\mathrm{OWA}}}}}
\newcommand{\owab}{{{{\beta}}}}

\newcommand{\pref}{\succ}

\begin{document}

\title{Finding a Collective Set of Items: From Proportional Multirepresentation to Group Recommendation\footnote{The preliminary version of this paper was presented at AAAI-2015.}}

\date{}

\author{
Piotr Skowron\\
       {University of Warsaw}\\
       {Warsaw, Poland}\\
\and
Piotr Faliszewski\\
       {AGH University}\\
       {Krakow, Poland}\\
\and
J\'{e}r\^{o}me Lang\\
       {Universit\'e Paris-Dauphine}\\
       {Paris, France}\\
}

\maketitle

\begin{abstract}
  We consider the following problem: There is a set of items (e.g.,
  movies) and a group of agents (e.g., passengers on a plane); each
  agent has some intrinsic utility for each of the items. Our goal is
  to pick a set of $K$ items that maximize the total derived utility
  of all the agents (i.e., in our example we are to pick $K$ movies
  that we put on the plane's entertainment system). However, the
  actual utility that an agent derives from a given item is only a
  fraction of its intrinsic one, and this fraction depends on how the
  agent ranks the item among the chosen, available, ones.
  We provide a formal specification of the model and provide concrete
  examples and settings where it is applicable. We show that the
  problem is hard in general, but we show a number of tractability
  results for its natural special cases.
\end{abstract}

\section{Introduction}

A number of real-world problems consist of selecting a set of items
for a group of agents to jointly use. Examples of such activities
include picking a set of movies to put on a plane's entertainment
system, deciding which journals a university library should subscribe
to, deciding what common facilities to build, or even voting for a
parliament (or other assembly of representatives). Let us consider some common features of these examples.

First, there is a set of items\footnote{We use the term `item' in the most
    neutral possible way. Items may be candidates running for
    an election, or movies, or possible facilities, and so on.} and a set of agents; each agent has some
  intrinsic utility for each of the items (e.g., this utility can be
  the level of appreciation for a movie, the average number of
  articles one reads from a given issue of a journal, expected benefit
  from building a particular facility, the feeling---measured in some way---of
  being represented by a particular politician).

  Second, typically it is not possible to provide all the items to the
  agents and we can only pick some $K$ of them, say (a plane's
  entertainment system fits only a handful of movies, the library has
  a limited budget, only several sites for the facilities are
  available, the parliament has a fixed size).

  Third, the intrinsic utilities for items extend to the sets of items
  in such a way that the utility derived by an agent from a given item
  may depend on the {\em rank} of this item (from the agent's point of
  view) among the selected ones.
  Extreme examples include the case where each agent derives utility
  from his or her most preferred item only (e.g., an agent will watch
  his or her favorite movie only, will read/use the favorite
  journal/favorite facility only, will feel represented by the most
  appropriate politician only), from his or her least preferred item
  only (say, the agent worries that the family will force him or her
  to watch the worst available movie), or derives $1/K$ of the utility
  from each of the available items (e.g., the agent chooses the
  item---say, a movie---at random).  However, in practice one should
  expect much more complicated schemes (e.g., an agent watches the top
  movie certainly, the second one probably, the third one perhaps,
  etc.; or, an agent is interested in having at least some $T$ interesting
  journals in the library; an agent feels represented by some top $T$
members of the parliament, etc.).

The goal of this paper is to formally define a model that captures all
the above-described scenarios, provide a set of examples where the
model is applicable, and  provide an initial set of computational
results for it in terms of efficient algorithms (exact or approximate)
and computational hardness results ($\np$-hardness and
inapproximability results).

Our work builds upon, generalizes, and extends quite a number of
settings that have already been studied in the literature.  We provide
a deeper overview of this research in Section~\ref{sec:related} and
here we only mention the two most directly related lines of work.
First, our model where the agents derive utility from their most
preferred item among the selected ones directly corresponds to winner
determination under the Chamberlin--Courant's voting
rule~\cite{ccElection,complexityProportionalRepr,fullyProportionalRepr}
(it is also very deeply connected to the model of budgeted social
choice
\cite{budgetSocialChoice,ore-luc:c:cc-online,bou-lu:c:value-directed-cc})
and is in a certain formal sense a variant of the facility location
problem. Second, the case where for each item each agent derives the
same fraction of the utility is, in essence, the same as $K$-winner
range-voting (or $K$-winner Borda~\cite{deb:j:k-borda}); that agents
enjoy equally the items they get is also a key assumption in the Santa
Claus problem~\cite{santa-claus}, and in the problem of designing
optimal picking sequences \cite{brams2000win,BL11,KalinowskiNW13}.

The paper is organized as follows. First, in Section~\ref{sec:model}
we discuss several important modeling choices and provide the formal
description of our model.  Then, in Section~\ref{sec:scenarios}, we
discuss the applicability of the model in various
scenarios. Specifically, we show a number of examples that lead to 
particular parameter values of our model.  We give an overview of our
results in Section~\ref{sec:overview} and then, in
Sections~\ref{sec:worst-case},~\ref{sec:approx-general},~and~\ref{sec:approx-borda},
we present these results formally.  In Section~\ref{sec:worst-case} we
present results regarding the complexity of computing exact solutions
for our model. In the next two sections we discuss the issue of
computing approximate solutions.  First without putting restrictions on
agents' utilities (Section~\ref{sec:approx-general}) and, then, for
what we call non-finicky utilities (Section~\ref{sec:approx-borda}).
Intuitively put, under non-finicky utilities the agents are required
to give relatively high utility values to a relatively large fraction
of the items).  We believe that the notion of non-finicky utilities is
one of the important contributions of this paper.  We discuss related
work in Section~\ref{sec:related} and conclude in
Section~\ref{sec:summary}.

\section{The Model}
\label{sec:model}

In this section we give a formal description of our model. However,
before we move on to the mathematical details, let us explain and
justify some high-level assumptions and choices that we have made.

First, we assume that the agents have separable preferences. This
means that the {\em intrinsic utility} of an object does not depend on
what other objects are selected. This is very different from, for
example, the case of combinatorial auctions.
However, in our model the {\em impact} of an object on the global
utility of an agent does depend on its rank (according to that agent)
among the selected items.
This distinction between the intrinsic value of an item and its value
distorted by its rank are also considered in several other research
fields, especially in decision theory (where it is known as ``rank-dependent utility theory'')
and in multicriteria decision making, from which we borrow one of the
main ingredients of our approach, the {\em ordered weighted average
  (OWA) operators}~\cite{Yager1988} (for technical details see the
work of Kacprzyk et al.~\shortcite{DBLP:series/sfsc/KacprzykNZ11}).
OWAs were recently used in social choice in several
contexts~\cite{GLMP14,ABLMR15,elk-ism:owa-cc}; we discuss these works
in detain in Section~\ref{sec:related}.


Second, throughout the paper we navigate between two views of the
agents' intrinsic utilities:
\begin{enumerate}
\item Generally, we assume that the utilities are provided explicitly
  in the input as numerical values, and that these values are
  comparable between agents. Yet, we make no further assumptions about
  the nature of agents' utilities: they do not need to be normalized,
  they do not need to come from any particular range of values,
  etc. Indeed, it is possible that some agent has very strong
  preferences regarding the items, modeled through high, diverse
  utility values, whereas some other agent does not care much about
  the selection process and has low utility values only.
\item In some parts of the paper (which will always be clearly
  identified), we assume that utilities are heavily constrained and
  derive from non-numerical information, such as approval ballots
  specifying which items an agent approves (leading to approval-based
  utilities), or rankings over alternatives, from which utilities are
  derived using an agent-independent scoring vector (typically, a
  Borda-like vector).
\end{enumerate}
Formally, the latter view is a special case of the former, but we
believe that it is worthwhile to consider it separately. Indeed, many
multiwinner voting rules (such as the
Chamberlin--Courant~\cite{ccElection} rule or the Proportional
Approval Voting rule~\cite{pavVoting}) fit the second view far more
naturally, whereas for other applications the former view is more
natural.

Third, we take the {\em utilitarian} view and measure the social
welfare of the agents as the sum of their perceived utilities.  One
could study other variants, such as the {\em egalitarian} variant,
where the social welfare is measured as the utility of the worst-off
agent. We leave this as possible future research (our preliminary
attempts indicated that the egalitarian setting is computationally
even harder than the utilitarian one).  Very recently, Elkind and
Isma{\"i}li~\cite{elk-ism:owa-cc} used OWA operators to define variants of
the Chamberlin--Courant rule that lay between the utilitarian and
egalitarian variants; we discuss this work in more detail in 
Section~\ref{sec:related}.\medskip



\subsection{The Formal Setting}
Let $N = [n]$ be a set of $n$ agents and let $A = \{a_1, \ldots,
a_m\}$ be a set of $m$ items. The goal is to pick a size-$K$ set $W$
of items that, in some sense, is most satisfying for the agents.
To this end,
%
%
(1) for each agent $i \in N$ and for each item $a_j \in A$ we have an
intrinsic utility $u_{i,a_j} \geq 0$ that agent $i$ derives from
$a_j$; (2) the utility that each agent derives from a set of $K$ items
is an ordered weighted average~\cite{Yager1988} of this agent's
intrinsic utilities for these items.

A weighted ordered average (OWA) operator over $K$ numbers is a
function defined through a vector $\alpha^{(K)} = \langle \alpha_1,
\ldots, \alpha_K\rangle$ of $K$ (nonnegative) numbers\footnote{The
  standard definition of OWAs assumes normalization, that is,
  $\sum_{i=1}^K \alpha_i = 1$. We do not make this assumption here,
  for the sake of convenience; note that whether OWA vectors are
  normalized or not is irrelevant to all notions and results of this
  paper.}
as follows. Let $\vec{x} = \langle x_1, \ldots, x_K \rangle$ be a
vector of $K$ numbers and let $\vec{x}^{\downarrow} = \langle
x^{\downarrow}_1, \ldots, x^{\downarrow}_K \rangle$
  be the nonincreasing
rearrangement of $\vec x$, that is, $x^{\downarrow}_i =
x_{\sigma(i)}$, where $\sigma$ is a permutation of $\{1,\ldots, K\}$
such that $x_{\sigma(1)} \geq x_{\sigma(2)} \geq \ldots \geq
x_{\sigma(K)}$.  Then we set:
\[
\textstyle \OWA_{\alpha^{(K)}}(\vec x) = \sum_{i=1}^K\alpha_i x^\downarrow_i
\]
To make the notation lighter, we write $\alpha^{(K)}(x_1, \ldots,
x_K)$, instead of $\OWA_{\alpha^{(K)}}(x_1, \ldots, x_K)$.

We provide a more detailed discussion of the OWA operators useful in
our context later and here we only mention that, for example, they can
be used to express the arithmetic average (through the size-$K$ vector
$(\frac{1}{K}, \ldots, \frac{1}{K})$), the maximum and minimum
operators (through vectors $(1,0, \ldots, 0)$, and $(0, \ldots, 0,1)$,
respectively) and the median operator (through the vector of all
zeros, with a single one in the middle position).

We formalize our problem of computing ``the most satisfying set of $K$
items'' as follows.

\begin{definition}\label{def:owa-winner}
  In the \textsc{OWA-Winner} problem we are given a set $N = [n]$ of
  agents, a set $A = \{a_1, \ldots, a_m\}$ of items, a collection of
  agent's utilities $(u_{i,a_j})_{i \in [n], a_j \in A}$, a positive
  integer $K$ ($K \leq m$), and a $K$-number OWA $\alpha^{(K)}$. The
  task is to compute a subset $W = \{w_1, \ldots, w_K\}$ of $A$ such
  that $u^{\alpha^{(K)}}_\util(W) = \sum_{i=1}^{n}
  \alpha^{(K)}(u_{i,w_1}, \ldots, u_{i,w_K})$ is maximal.
\end{definition}



\begin{example}\label{example}
  Consider six agents with the following utilities over the
  items from the set $A = \{a_1, a_2, a_3, a_4, a_5, a_6
  \}$: 

  \[
    \begin{array}{c|cccccc}
             & u(a_1) & u(a_2) & u(a_3) & u(a_4) & u(a_5) & u(a_6)\\ 
    \hline
    \mathtt{3\ agents} & 5 & 4 & 3 & 0 & 2 & 1\\  
    \mathtt{2\ agents} & 4 & 0 & 2 & 3 & 1 & 5\\  
    \mathtt{1\ agent}  & 0 & 3 & 2 & 4 & 5 & 1
  \end{array}
\]
We want to select $K=3$ items and we use OWA $\alpha = (2,1,0)$.  What
is the score of $\{a_1,a_2,a_6\}$? The first three agents get utility
$2 \times 5 + 4 = 14$ each, the next two get $2 \times 5 + 4 = 14$
each, and the last one gets $2 \times 3 + 1 = 7$. So, the score of
$\{a_1,a_2,a_6\}$ is $42 + 28 + 7 = 77$. Indeed, this is the optimal
set; the next best ones are $\{a_1,a_2,a_4\}$, $\{a_1,a_2,a_5\}$ and
$\{a_1,a_5,a_6\}$, all with score $75$.  The rule defined by the OWA
$\alpha' = (1,1,1)$, known as $3$-Borda (due to the very specific
values of agents' utilities; see Example~\ref{example:borda} in the
next section), would choose $\{a_1,a_2,a_3\}$ and
Chamberlin--Courant's rule (in our terms, the rule defined by the OWA
operator $\alpha'' = (1,0,0)$) would choose $\{a_1,a_5,a_6\}$.
\end{example}

For a family $(\alpha^{(K)})_{K=1}^\infty$ of OWAs, we write
$\alpha$-\textsc{OWA-Winner} to denote the variant of the problem
where for each given solution size $K$ we use OWA $\alpha^{(K)}$.
From now on we will not mention the size of the OWA vector explicitly
and it will always be clear from the context.  We implicitly assume
that OWAs in our families are polynomial-time computable.

\subsection{Classes of Intrinsic Utilities}\label{sec:utilities}
While our general setting allows agents to express arbitrary
utilities, we also focus on two cases where they only provide
dichotomous or ordinal information:
\begin{description}
\item[Dichotomous information.] Agents provide {\em dichotomous
    information} if they only have to specify which items they like.
  This information is then mapped into {\em dichotomous} (or, as we
  typically refer to them, {\em approval-based}) utilities, defined by
  $u_i(a_j) = 1$ if $i$ likes $a_j$ and $u_i(a_j) = 0$ otherwise.
\item[Ordinal information.] Agents provide {\em ordinal information}
  if they only have to specify their rankings over items, called their
  {\em preference orders}. This information is then mapped into
  utilities using a scoring vector, exactly in the same way as
  positional scoring rules (for single-winner voting) do. We focus on
  the partiuclar case where this scoring vector is the {\em Borda}
  vector, i.e., if the rank of $a_j$ in $i$'s ranking is $k$ then
  $u_i(a_j) = m-k$. We refer to this setting as \emph{Borda-based}
  utilities.
\end{description}
Naturally, these are special cases of our general setting. Yet using
approval-based or Borda-based utilities can be more convenient than
using the general approach.

%
%
%
%

\begin{example}\label{example:borda}
  The utilities of the agents from Example~\ref{example} are
  Borda-based and can be expresses as the following preference orders:
  \begin{align*}
    \mathtt{3\ agents} \colon & a_1 \pref a_2 \pref a_3 \pref a_5 \pref a_6 \pref a_4 \\
    \mathtt{2\ agents} \colon & a_6 \pref a_1 \pref a_4 \pref a_3 \pref a_5 \pref a_2 \\
    \mathtt{1\ agent}  \colon & a_5 \pref a_4 \pref a_2 \pref a_3 \pref a_6 \pref a_1 
  \end{align*}
\end{example}

Both approval-based utilities and Borda-based utilities are inspired
by analogous notions from the theory of voting, where approval and
Borda count are very well-known single-winner voting rules (briefly
put, under these rules we treat the utilities of the items as their
scores, sum up the scores assigned to the items by the voters, and
elect the item that has the highest score). Further, Borda-based
utilities have been used in the original Chamberlin--Courant's rule
and in several works on fair division (see, e.g., a paper of Brams and
King~\shortcite{brams2005efficient}).

One of the high-level messages of this paper is that
\textsc{OWA-Winner} problems tend to be computationally easier for the
case of Borda-based utilities than for the case of approval-based ones
(while we typically obtain $\np$-hardness in both settings, we find
good approximation algorithms for many of the Borda-based cases,
whereas for the approval-based setting our algorithms are either
significantly weaker or we obtain outright inapproximability results).
This is so mostly because under Borda-based utilities all the agents
assign relatively high utility values to a relatively large fraction
of items. In the following definition we try to capture this property.

\begin{definition}
  Consider a setting with $m$ items and let $u_{\max}$ denote the
  highest utility that some agent gives to an item. Let $\beta$ and
  $\gamma$ be two numbers in $[0,1]$. We say that the agents have
  ($\beta$, $\gamma$)-non-finicky utilities if every agent has utility
  at least $\beta u_{\max}$ for at least $\gamma m$ items.
\end{definition}

To understand this notion better, let us consider the following
example.
\begin{example}\label{exjl}
  Let $n = 3$ and $m = 6$. The utilities are as defined below:
\[
  \begin{array}{c|cccccc}
             & u(a_1) & u(a_2) & u(a_3) & u(a_4) & u(a_5) & u(a_6)\\ 
    \hline
    \mathtt{Agent\ 1} & 10 & 10 & 9 & 8 & 5 & 0\\  
    \mathtt{Agent\ 2} & 6 & 5 & 0 & 10 & 8 & 10\\  
    \mathtt{Agent\ 3} & 8 & 0 & 10 & 6 & 10 & 7
  \end{array}
\]
The agents have $(0.8, 0.5)$-non-finicky utilities. Indeed, all there
agents have utility at least 8 for at least half of the items. They
also have $(0.6, \frac{2}{3})$-non-finicky utilities, and $(0.5,
\frac{5}{6})$-non-finicky utilities.  We will also use the agents and
items from this example later, when presenting our algorithms.
\end{example}

As we can expect, Borda-based utilities are non-finicky in a very
natural sense.
\begin{observation}
  For every $x$, $0 \leq x \leq 1$, Borda-based utilities are $(x,
  1-x)$-non-finicky.
\end{observation}

However there are also other natural cases of non-finicky
utilities. For example, consider agents that have approval-based
utilities and where each agent approves of at least a $\gamma$
fraction of the items. These agents have $(1,\gamma)$-non-finicky
utilities.  (The reader may be surprised here that approval-based
utilities may be non-finicky even though we said that we obtain
inapproximability results for them. Yet, there is no contradiction
here. These inapproximability results rely on the fact that some
agents approve of very few items.)

\subsection{A Dictionary of Useful OWA Families}\label{sec:dictionary}

Below we give a catalog of OWA families that we focus on throughout
the paper (in the description below we take $K$ to be the dimension of
the vectors to which we apply a given OWA).

\begin{enumerate}
\item \textbf{$\boldsymbol{k}$-median OWA.} For each $k \in \{1,\ldots,
  K\}$, $k\text{-}\med$ is the OWA defined by the vector of $k-1$
  zeros, followed by a single one, followed by $K-k$ zeros.  It is
  easy to see that $k\text{-}\med (x_1, \ldots, x_K)$ is the $k$-th
  largest number in the set $\{x_1, \ldots, x_K\}$ and is known as the
  $k$-median of $\vec x$. In particular, $1\text{-}\med(\vec x)$ is
  the maximum operator, $K\text{-}\med(\vec x)$ is the minimum
  operator, and if $K$ is odd, $\frac{K+1}{2}\text{-}\med(\vec x)$ is
  the median operator.

\item \textbf{$\boldsymbol{k}$-best OWA.} For each $k \in \{1,\ldots,
  K\}$, $k\text{-}\best$ OWA is defined through the vector of $k$ ones
  followed by $K-k$ zeros. That is, $k\text{-}\best({\vec x})$ is the
  sum of the top $k$ values in $\vec x$ (with appropriate scaling,
  this means an arithmetic average of the top $k$ numbers).
  $K\text{-}\best_K$ is simply the sum of all the numbers in $\vec x$
  (after scaling, the arithmetic average).

\item \textbf{Arithmetic progression OWA.} This OWA is defined through
  the vector $\ap = \langle a+(K-1)b, a+(K-2)b, \ldots, a \rangle$,
  where $a \geq 0$ and $b > 0$.  (One can easily check that the choice
  of $b$ has no impact on the outcome of OWA-Winner; this is not the
  case for $a$, though.)

\item \textbf{Geometric progression OWA.} This OWA is defined through
  the vector $\gp = \langle p^{K-1}, p^{K-2}, \ldots, 1 \rangle$,
  where $p > 1$. (This is without loss of generality, because multiplying
  the vector by a constant factor has no  impact on the outcome
  of OWA-Winner; but the choice of $p$ matters.)

\item \textbf{Harmonic OWA.} This OWA is defined through the vector
  $\langle 1, \frac{1}{2}, \frac{1}{3}, \ldots, \frac{1}{K} \rangle$,

\item \textbf{Hurwicz OWA.} This OWA is defined through a vector
  $(\lambda, 0, \ldots, 0, 1-\lambda)$, where $\lambda$, $0 \leq
  \lambda \leq 1$, is a parameter.
\end{enumerate}

Naturally, all sorts of middle-ground OWAs are possible between these
particular cases, and can be tailored for specific applications.  As
our natural assumption is that highly ranked items have more impact
than lower-ranked objects, we often make the assumption that OWA
vectors are {\em nonincreasing}, that is,
 $\alpha_1 \geq \ldots \geq \alpha_K$. 
While most OWA operators we consider in the paper are indeed nonincreasing, this is
not the case for $k$-medians (except for $1$-median) and Hurwicz 
(except for $\lambda = 1$).

\section{Applications of the Model}\label{sec:scenarios}
We believe that our model is very general. To substantiate this claim,
in this section we provide four quite different scenarios where it is
applicable.  \medskip

\noindent\textbf{Generalizing Voting Rules.}\quad
Our research started as an attempt to generalize the rule of
Chamberlin and Courant~\shortcite{ccElection} for electing sets of
representatives.
For this rule, the voters (the agents) have Borda-based utilities over
a set of candidates and we wish to elect a $K$-member committee (e.g.,
a parliament), such that each voter is represented by one member of
the committee.  If we select $K$ candidates, then a voter is
``represented'' by the selected candidate that she ranks highest
among the chosen ones.
Thus, winner determination under Chamberlin--Courant's voting rule
boils down to solving $1\text{-}\mathrm{best}$-\textsc{OWA-Winner} for
the case of Borda-based utilities.  On the other hand, solving
$K\text{-}\mathrm{best}$-\textsc{OWA-Winner} for Borda-based utilities
is equivalent to finding winners under $K$-Borda, the rule that picks
$K$ candidates with the highest Borda scores (see the work of Elkind
et al.~\shortcite{elk-fal-sko-sli:c:multiwinner-rules} for a
classification of multiwinner voting rules, including, e.g., $K$-Borda
and Chamberlin--Courant's rule).

Our model extends one more appealing voting rule, known as
Proportional Approval Voting (PAV; see the work of
Kilgour~\shortcite{pavVoting} for a review of approval-based
multiwinner rules, and the work of Aziz et
al.~\shortcite{azi-gas-gud-mac-mat-wal:c:approval-multiwinner} and
Elkind and Lackner~\cite{elk-lac:c:dichotomous-prefs} for
computational results). Winner determination under PAV is equivalent
to solving $\alpha$-\textsc{OWA-Winner} for the harmonic OWA, for the
case of approval-based utilities.  \medskip

\noindent\textbf{Malfunctioning Items or Unavailable Candidates.}\quad
Consider a setting where we pick the items off-line, but on-line it
may turn out that some of them are unavailable (for example, we pick a
set of journals the library subscribes to, but when an agent goes to a
library, a particular journal could already be borrowed by someone
else; see the work of Lu and Boutilier~\shortcite{LuBoutilier10} for
other examples of social choice with possibly unavailable candidates).
We assume that each item is available with the same, given,
probability $p$ (i.i.d.). The utility an agent gets from a set of
selected items $W$ is the expected value of the best available
object. The probability that the $i$'th item is available while the
preceding $i-1$ items are not, is proportional to $p(1-p)^{i-1}$. So,
to model the problem of selecting items in this case, we should use
the geometric progression OWA with initial value $p$ and coefficient
$1-p$.
\medskip

\noindent\textbf{Uncertainty Regarding How Many Items a User Enjoys.}\quad 
There may be some uncertainty about the number of items a user would
enjoy (e.g., on a plane, it is uncertain how many movies a passenger
would watch; one might fall asleep or might only watch those
movies that are good enough). We give two possible models for the
choice of the OWA vectors:

\begin{enumerate}\fixlist
\item The probability that an agent enjoys $i$ items, for $0 \leq i
  \leq K$, is uniformly distributed, i.e., an agent would enjoy
  exactly his or her first $i$ items in $W$ with probability
  $\frac{1}{K+1}$. So, the agent enjoys the $i$'th item if she enjoys
  at least $i$ items, which occurs with probability
  $\frac{K-i+1}{K+1}$; we should use OWA vector defined by
  $\alpha_i = K-i+1$ (we disregard the normalizing constant), i.e., an
  arithmetic progression. 
\item We assume that the values given by each user to each item are
  distributed uniformly, i.i.d., on $[0,1]$ and that each user uses
  only the items that have a value at least $\theta$, where $\theta$
  is a fixed (user-independent) threshold. Therefore, a user enjoys
  the item in $W$ ranked in position $i$ if she values at least
  $i$ items at least $\theta$, which occurs with probability
  $\sum_{j=i}^K{K \choose i} (1-\theta)^i\theta^{K-i}$, thus leading
  to the OWA vector defined by $\alpha_i = \sum_{j=i}^K{K \choose i}
  (1-\theta)^i\theta^{K-i}$.
\end{enumerate}

\noindent\textbf{Ignorance About Which Item Will Be Assigned to a User.}\quad
We now assume that a matching mechanism will be used after selecting
the $K$ items. The matching mechanism is not specified; it might also
be randomized. If the agents have a {\em complete ignorance about the
  mechanism used}, then it makes sense to use known criteria for
decision-making under complete uncertainty:
\begin{enumerate}\fixlist
\item The {\em Wald} criterion assumes that agents are extremely
  risk-averse, and corresponds to $\alpha = K\text{-}\med = \langle 0, \ldots, 0, 1 \rangle$.  The
  agents consider their worst possible items.
\item The {\em Hurwicz} criterion is a linear combination between the
  worst and the best outcomes, and corresponds to $\alpha = (\lambda,
  0, \ldots, 0, 1-\lambda)$ for some fixed $\lambda \in (0,1)$.
\end{enumerate}
If the agents know that they are guaranteed to get one of their best
$i$ items, then the Wald and Hurwicz criteria lead, respectively, to
the OWAs $\alpha = i\text{-}\med$ and $\alpha = (\lambda, 0, \ldots,
0, 1-\lambda, 0, \ldots, 0)$, with $1-\lambda$ in position $i$.
If the agents know that the mechanism gives them one of their top
$i$ items, each with the same probability, then we should use
$i\text{-best}$ OWA.  More generally, the matching mechanism may
assign items to agents with a probability 
that decreases when the rank increases.

\section{Overview of the Results}
\label{sec:overview}

\begin{table}[tb!]
  \footnotesize
\centering
\begin{tabular}{p{3.2cm}|c|c|l}
\multicolumn{4}{c}{}\\
             & general and       & $(\beta, \gamma)$--non-finicky & \\
  OWA family & approval utilities & and Borda utilities & References\\

\hline&&&\\[-0.7em]

\multirow{2}{*}{$k$-median ($k$ fixed)}      
                            & $\np$-hard & $\np$-hard (Borda)   & Proposition~\ref{thm:k-best-k-med} \\
                            & \textsc{DkS}-bounded & $(\beta -\epsilon)$-approx. & Theorem~\ref{thm:dks}~and Corollary~\ref{cor:borda-ptas} \\
                            &                      & PTAS (Borda) & Theorem~\ref{thm:borda:any-first-ell-ptas} \\[0.6em]

\multirow{2}{*}{$K$-median}   
                            & $\np$-hard & $\np$-hard & Theorems~\ref{thm:k-1-best}~and~\ref{thm:borda:k-1-best}\\
                            & \textsc{MEBP}-bounded & ? & Theorem~\ref{thm:mebp-bounded}, open problem \\[0.2em]

  \hline
  &&&\\[-0.7em]
  \multirow{2}{*}{$1$-best}   & $\np$-hard (approval) & $\np$-hard (Borda)  & Literature~\cite{complexityProportionalRepr, budgetSocialChoice}\\
  & $(1-\frac{1}{e})$-approx. & $(\beta -\epsilon)$-approx.     & Literature~\cite{budgetSocialChoice}, Corollary~\ref{cor:borda-ptas}\\ 
  &                           & PTAS (Borda)          & Literature~\cite{sko-fal-sli:c:multiwinner}  \\[0.6em] 

\multirow{2}{*}{$k$-best ($k$ fixed)}      
                            & $\np$-hard (approval) & $\np$-hard (Borda)    & Proposition~\ref{thm:k-best-k-med} \\
                            & $(1-\frac{1}{e})$-approx. & $(\beta -\epsilon)$-approx. & Theorem~\ref{thm:greedyAprox}~and Corollary~\ref{cor:borda-ptas} \\[0.6em]

\multirow{2}{*}{$(K-1)$-best}   
                            & $\np$-hard (approval) & $\np$-hard (Borda) & Theorems~\ref{thm:k-1-best}~and~\ref{thm:borda:k-1-best}\\
                            & PTAS & PTAS & Theorem~\ref{thm:approx} \\[0.6em]

{$K$-best}   & $\p$ & $\p$ & folk result \\[0.2em]

\hline&&&\\[-0.7em]

\multirow{2}{*}{arithmetic progression}      
                            & $\np$-hard & ?    & Theorem~\ref{thm:almost-all-hard}, open problem\\
                            & $(1-\frac{1}{e})$-approx. & $(1-\frac{1}{e})$-approx. & Theorem~\ref{thm:greedyAprox} \\[0.6em]

\multirow{2}{*}{geometric progression}      
                            & $\np$-hard & ?    & Theorem~\ref{thm:almost-all-hard}, open problem\\
                            & $(1-\frac{1}{e})$-approx. & $\beta -\epsilon$ & Theorem~\ref{thm:greedyAprox}, Corollary~\ref{cor:borda-geometric-ptas} \\[0.2em]

\hline&&&\\[-0.7em]

\multirow{3}{*}{Hurwicz[$\lambda$]}      
                            & $\np$-hard (approval) & ?    & Corollary~\ref{cor:hurwicz:np-hard}, open problem\\
                            & $\lambda(1-\frac{1}{e})$-approx. & $\lambda(1-\epsilon)$-approx. & Corollary~\ref{cor:hurwicz} \\
                            & & for each $\epsilon > 0$ & \\

\end{tabular}
\caption{Summary of our results for the OWA families 
  from Section~\ref{sec:dictionary}. For each OWA family we provide four 
  entries: In the first row (for a given OWA family) we give its worst 
  case complexity (in the general case and in the non-finicky utilities 
  case), and in the second row we list the best known approximation result 
  (in the general case and in the non-finicky utilities case). We write $K$ 
  to mean the cardinality of the winner set that we seek. In the ``References'' 
  column we point to the respective result in the paper/literature. For 
  negative results we indicate the simplest types of utilities where they 
  hold; for positive results we give the most general types of utilities 
  where they hold. For approximability results for the case of non-finicky 
  utilities, we write $(\beta-\epsilon)$-approx to mean that there is a 
  polynomial-time approximation algorithm whose approximation ratio approaches 
  $\beta$ as the size of the committee increases (in effect, for each $\epsilon$, 
  $\epsilon > 0$, there is a polynomial-time algorithm that achieves 
  $\beta-\epsilon$ approximation ratio, by using a brute-force algorithm is 
  the size of the committee is smaller than a certain constant). For inapproximability
  results, by \textsc{DkS}-bounded and \textsc{MEBP}-bounded we mean, 
  respectively, inapproximability results derived from the \textsc{Densest-k-Subgraph} 
  problem and from the \textsc{Maximum Edge Biclique Problem}.}
\label{tab:summary}
\end{table}

In this section we provide a high-level overview of our results. It
turns out that computational properties of the \textsc{OWA-Winner}
problem are quite varied and strongly depend on the types of OWA
operators and the allowed agent utilities. We present a summary of our
results in Table~\ref{tab:summary} (however, we stress that some of
our technical results are not listed in the table and can be found
only in the following sections).

Our first observation is that without any restrictions,
\textsc{OWA-Winner} is $\np$-hard. This is hardly surprising since the
problem generalizes other $\np$-hard problems, and it is natural to
ask if there are any special cases where it is easy. Unfortunately, as
we show in Section~\ref{sec:worst-case}, they are very rare.  For
example, without restrictions on the agents' utilities,
\textsc{OWA-Winner} can be solved in polynomial time either if we
treat $K$ as a constant or if we use the constant OWA vector (i.e., if
we use $K\text{-}\mathrm{best}$ OWA). Indeed, the problem becomes
$\np$-hard already for the $(K-1)\text{-}\mathrm{best}$ OWA.  This
holds even if the agents are restricted to have approval-based
utilities (Theorem~\ref{thm:k-1-best}) or Borda-based utilities
(Theorem~\ref{thm:borda:k-1-best}).  More generally, we show that
\textsc{OWA-Winner} is $\np$-hard for every family of OWA vectors that
are nonconstant and nonincreasing (Theorem~\ref{sec:worst-case}),
which captures a significant fraction of all interesting settings.

After considering the worst-case complexity of computing exact
solutions in Section~\ref{sec:worst-case}, in
Section~\ref{sec:approx-general} we focus on the approximability of
the \textsc{OWA-Winner} problem. We show that in this respect there is
a significant difference between two main classes of OWA vectors,
those that are nonincreasing and the remaining ones. We show that for
the nonincreasing OWA vectors the standard greedy algorithm for
optimizing submodular functions achieves approximation ratio of
($1-1/e$), irrespective of the nature of the agents' utilities
(Lemma~\ref{lemma:constOWA} and Theorem~\ref{thm:greedyAprox}). On the
other hand, we present evidence that there is little hope for good
approximation algorithms for the case of OWA vectors that are not
nonincreasing (Example~\ref{example:non-sub} and
Theorems~\ref{thm:dks}~and~\ref{thm:mebp-bounded}).

Next, in Section~\ref{sec:approx-borda}, we consider approximation
algorithms for \textsc{OWA-Winner} for the case where agents have
non-finicky utilities. It turns out that for non-finicky utilities we
can sometimes obtain much better approximability guarantees than in
the general setting. The key feature of non-finicky utilities
assumption is that every agent gives sufficiently high utility values
to sufficiently many items, so that the algorithms have enough
flexibility in picking the items to achieve high quality results.
Specifically, we show a strong approximation algorithm for the case of
non-finicky utilities and OWA vectors that concentrate most of the
weight in a constant number of their top coefficients
(Theorems~\ref{thm:borda:nonincreasing},~\ref{thm:borda:any-first-ell},~\ref{thm:borda:any-first-ell-ptas},
and Corollary~\ref{cor:borda-geometric-ptas}). These results apply,
for example, to the case of geometric progression OWAs,
$\ell\text{-}\mathrm{best}$ OWAs, and $\ell\text{-}\mathrm{med}$ OWAs
(for fixed values of $\ell$).  Further, when applied to the case of
Borda-based utilities (which, as we have argued in
Section~\ref{sec:utilities}, are non-finicky in a very strong sense),
we obtain polynomial-time approximation schemes (that is,
approximation algorithms that can compute solutions with an
arbitrarily good precision, but whose running time depends
polynomially only on the size of the problem but not necessarily on
the desired approximation ratio).

\section{Computing Exact Solutions}\label{sec:worst-case}

We start our analysis by discussing the complexity of solving the
\textsc{OWA-Winner} problem exactly. In general, it seems that
\textsc{OWA-Winner} is a rather difficult problem and below we show
this section's main negative result. That is, we show that our problem
is $\np$-hard for \emph{any class of OWA vectors} satisfying a certain
natural restriction. Intuitively, this restriction says that in a
considered family of OWAs, the impact of more-liked items on the total
satisfaction of an agent is greater than that of the less-liked ones.

\begin{theorem}\label{thm:almost-all-hard}
  Fix an OWA family $\alpha$ such that for every $K$, $\alpha^{(K)}$
  is nonincreasing and nonconstant. $\alpha$-\textsc{OWA-Winner} is
  $\np$-hard, even for approval-based utilities.
\end{theorem}

For the sake of readability, we first prove two simpler results that
we later use in the proof of Theorem~\ref{thm:almost-all-hard}.  In
these proofs, we give reductions from the standard
\textsc{VertexCover} problem and from \textsc{CubicVertexCover}, its
variant restricted to cubic graphs.

\begin{definition}
  In the \textsc{VertexCover} problem we are given an undirected graph
  $G = (V,E)$, where $V = \{v_1,\ldots, v_m\}$ is the set of vertices
  and $E = \{e_1, \ldots, e_n\}$ is the set of edges, and a positive
  integer $K$. We ask if there is a set $C$ of up to $K$ vertices such
  that each edge is incident to at least one vertex from $C$.
%
  The \textsc{CubicVertexCover} problem the same problem, restricted to graphs
  where each vertex has  degree exactly three.
\end{definition}

\textsc{VertexCover} is well-known to be
$\np$-hard~\cite{gar-joh:b:int}; $\np$-hardness for
\textsc{CubicVertexCover} was shown by Alimonti and
Kann~\shortcite{journals/tcs/AlimontiK00}.

\begin{theorem}\label{thm:constantThresholdNPHard}
  Fix an OWA family $\alpha$, such that there exists $p$ such that for
  every $\alpha^{(K)}$ we have $\alpha^{(K)}_p >
  \alpha^{(K)}_{p+1}$. $\alpha$-\textsc{OWA-Winner} is $\np$-hard,
  ever for approval-based utilities.
\end{theorem}
\begin{proof}
  We give a reduction from \textsc{CubicVertexCover} problem.  Let $I$
  be an instance of \textsc{CubicVertexCover} with graph $G = (V,E)$,
  where $V = \{v_1, \dots, v_m\}$ and $E = \{e_1, \dots, e_n\}$, and
  positive integer $K$. W.l.o.g., we assume that $n > 3$.

  We construct an instance $I'$ of $\alpha$-\textsc{OWA-Winner}. In
  $I'$ we set $N = E$ (the agents correspond to the edges), $A = V
  \cup \{b_1, b_2 ,\dots b_{p-1}\}$ (there are $(p-1)$ dummy items;
  other items correspond to the vertices), and we seek a collection of
  items of size $K+p-1$. 
  Each agent $e_i$, $e_i \in E$, has utility $1$ exactly for all the
  dummy items and for two vertices that $e_i$ connects and for each of
  the dummy items (for the remaining items $e_i$ has utility $0$).  In
  effect, each agent has utility $1$ for exactly $p+1$ items.



  We claim that $I$ is a yes-instance of \textsc{CubicVertexCover} if
  and only if there exists a solution for $I'$ with the total utility
  at least $n\sum_{i = 1}^{p}\alpha_i + (3K - n)\alpha_{p+1}$.

  $(\Rightarrow)$ If there is a vertex cover $C$ of size $K$ for $G$,
  then by selecting the items $W = C \cup \{b_1, b_2 ,\dots b_{p-1}\}$
  we obtain the required utility of the agents. Indeed, for every
  agent $e_i$ there are at least $p$ items in $W$ for which $i$ gives
  value $1$ (the $p-1$ dummy items and at least one vertex incident to
  $e_i$). These items contribute the value $n\sum_{i = 1}^{p}\alpha_i$
  to the total agents' utility. Additionally, since every non-dummy
  item has value $1$ for exactly 3 agents, and since every agent has
  at most $(p+1)$ items with value $1$, there are exactly $(3K - n)$
  agents that have exactly $(p+1)$ items in $W$ with values $1$.
  These $(p+1)$'th additional utility-$1$ items of the $(3k-n)$ agents
  contribute $(3K - n)\alpha_{p+1}$ to the total utility.  Altogether,
  the agents' utility is $n\sum_{i = 1}^{p}\alpha_i + (3K -
  n)\alpha_{p+1}$, as claimed.

  $(\Leftarrow)$ Let us assume that there is a set of $(K+p-1)$ items
  with total utility at least $n\sum_{i = 1}^{p}\alpha_i + (3K -
  n)\alpha_{p+1}$. In $I'$ we have $(p-1)$ items that have value $1$
  for each of the $n$ agents, and every other item has value $1$ for
  exactly $3$ agents. Thus, the sum of the utilities of $K+p-1$ items
  (without applying the OWA operator yet) is at most $(p-1)n + 3K = pn
  + (3K - n)$.  Thus, the total utility of the agents (now applying
  the OWA operator) is $n\sum_{i = 1}^{p}\alpha_i + (3K -
  n)\alpha_{p+1}$ only if for each agent $e_i$ the solution contains
  $p$ items with utility $1$. Since there are only $p-1$ dummy items,
  it means that for each agent $e_i$ there is a vertex $v_j$ in the
  solution such that $e_j$ is incident to $v_j$. That is, $I$ is a
  yes-instance of \textsc{CubicVertexCover}.
\end{proof}

\begin{theorem}\label{thm:k-1-best}
  $(K-1)\text{-}\mathrm{best}$-\textsc{OWA-Winner} is $\np$-complete even for approval-based
  utilities.
\end{theorem}
\begin{proof}
  Membership in $\np$ is clear.  We show a reduction from the
  \textsc{VertexCover} problem. 
%
%
  Let $I$ be an instance of \textsc{VertexCover} with graph $G =
  (V,E)$, where $V = \{v_1, \dots, v_m\}$ and $E = \{e_1, \dots,
  e_n\}$, and with a positive integer $K$ (without loss of generality,
  we assume that $K \geq 3$ and $K < m$).


  We construct an instance $I'$ of $(K-1)$-best-\textsc{OWA-Winner} in
  the following way.  We let the set of items be $A = V$ and we form
  $2n$ agents, two for each edge. Specifically, if $e_i$ is an edge
  connecting two vertices, call them $v_{i,1}$ and $v_{i,2}$, then we
  introduce two agents, $e_i^1$ and $e_i^2$, with the following
  utilities: $e_i^1$ has utility $1$ for $v_{i, 1}$ and for $v_{i,
    2}$, and has utility $0$ for all the other items; $e_i^2$ has
  opposite utilities---it has utility $0$ for $v_{i, 1}$ and for
  $v_{i, 2}$, and has utility $1$ for all the remaining ones.

  Let $W$ be some set of $K$ items (i.e., vertices) and
  consider the sum of the utilities derived by the two agents $e_i^1$
  and $e_i^2$ from $W$ under $(K-1)$-best-OWA. If neither $v_{i, 1}$
  nor $v_{i, 2}$ belong to $W$, then the total utility of $e_i^1$ and
  $e_i^2$ is equal to $K-1$ (the former agent gets utility $0$ and the
  latter one gets $K-1$). If only one of the items, i.e.,
  either $v_{i, 1}$ or $v_{i, 2}$, belongs to $W$, then the total
  utility of $e_i^1$ and $e_i^2$ is equal to $K$ (the former agent
  gets utility $1$ and the latter one still gets $K-1$). Finally, if
  both items $v_{i, 1}, v_{i, 2}$ belong to $W$, then the total
  utility of $e_i^1$ and $e_i^2$ is also equal to $K$ (the former gets
  utility $2$ and the latter gets utility $K-2$). Thus the total
  utility of all agents is equal to $K \cdot n$ if and only if the
  answer to the instance $I$ is ``yes''. This shows that the reduction
  is correct and, since the reduction is computable in polynomial
  time, the proof is complete.
\end{proof}

Using a proof that combines the ideas of the proof of
Theorems~\ref{thm:constantThresholdNPHard} and~\ref{thm:k-1-best}, we show that
indeed \textsc{OWA-Winner} is $\np$-hard for a large class of natural
OWAs. \medskip

\newenvironment{proof-sketch-almost-all-hard}{\noindent\textbf{\textit{Proof
      of Theorem~\ref{thm:almost-all-hard}}.}\quad}{\hfill\medskip}

\begin{proof}[Proof of Theorem~\ref{thm:almost-all-hard}]
  We give a reduction from \textsc{CubicVertexCover}.  Let $I$ be an instance of
  \textsc{CubicVertexCover} with graph $G = (V,E)$, where $V = \{v_1,
  \dots, v_m\}$ and $E = \{e_1, \dots, e_n\}$, and with positive
  integer $K$.

  Now let us consider $\alpha^{(2K)}$. Since $\alpha^{(2K)}$ is
  nonincreasing and nonconstant, one of the two following conditions
  must hold.
  \begin{enumerate}
  \item There exists $p\leq K$ such that $\alpha^{(2K)}_p >
    \alpha^{(2K)}_{p+1}$.
  \item There exists $p > K$ such that $\alpha^{(2K)}_p >
    \alpha^{(2K)}_{p+1}$, and for every $p \leq K$, we have
    $\alpha^{(2K)}_p = \alpha^{(2K)}_{p+1}$.
  \end{enumerate}

  If (1) is the case then we use a reduction similar to that in the
  proof of Theorem~\ref{thm:constantThresholdNPHard}. The only
  difference is that apart from the set $D_1$ of $(p-1)$ dummy items
  (ranked first by all agents), we introduce the set $D_2$ of
  $(2K-p+1)$ dummy items and $(2K-p+1)$ sets $N_1, N_2, \dots,
  N_{2K-p+1}$, each consisting of $2n$ dummy agents. The dummy items
  from $D_2$ are introduced only to fill-up the solution up to $2K$
  members. The dummy agents from $N_i$ have utility $1$ for each of
  the items from $D_1$ and for the $i$'th item from $D_2$ (they have
  utility $0$ for all the other items). This is to enforce that the
  items from $D_2$ are selected in the optimal solution. The further
  part of the reduction is as in the proof of
  Theorem~\ref{thm:constantThresholdNPHard}.

  If (2) is the case, then we use a reduction similar to that in the
  proof of Theorem~\ref{thm:k-1-best}. We let the set of items be $A =
  V \cup D_1 \cup D_2$, where $D_1$, $|D_1| = p+1-K$, and $D_2$,
  $|D_2| = 2K-p-1$ are sets of dummy items that we need for our
  construction. Similarly as in the proof of
  Theorem~\ref{thm:k-1-best}, for each edge $e_i \in E$ we introduce
  two agents $e_i^1$ and $e_i^2$. Here, however, we additionally need
  the set $F$ of $(2n+1)$ dummy agents. Each dummy agent from $F$
  assigns utility 1 to each dummy item from $D_2$ and utility 0 to the
  remaining items---consequently, since $|F| > 2n$, each dummy item
  from $D_2$ must be selected to every optimal solution. Further, each
  non-dummy agent assigns utility 1 to each dummy agent from
  $D_1$---this way we ensure that every item from $D_1$ must be
  selected to every optimal solution. Finally, the utilities of the
  non-dummy agents for the non-dummy items are defined exactly as in
  the proof of Theorem~\ref{thm:k-1-best}. This ensures that the
  optimal solution, apart from $D_1$ and $D_2$, will contain the
  non-dummy items that correspond to the vertices from the optimal
  vertex cover.
\end{proof}

One may wonder if our just-presented hardness results also hold for
other restrictions on agents' utilities.  Below we show a variant of
the result from Theorem~\ref{thm:k-1-best} for Borda-based
utilities. It follows by an application of a similar idea as in the
proof of Theorem~\ref{thm:k-1-best}, but the restriction to
Borda-based utilities requires a much more technical proof (available
in the appendix).


\newcommand{\thmbordakminusonebest}{
  $(K-1)\text{-}\mathrm{best}$-\textsc{OWA-Winner} is $\np$-hard even
  for Borda-based utilities.  }

\begin{theorem}\label{thm:borda:k-1-best}
  \thmbordakminusonebest
\end{theorem}

\subsection{Inherited Hardness Results}

We now consider the cases of
$k\text{-}\mathrm{best}$-\textsc{OWA-Winner} and
$k\text{-}\mathrm{med}$-\textsc{OWA-Winner} (where $k$ is a
constant). By results of Procaccia, Rosenschein and
Zohar~\shortcite{complexityProportionalRepr} and Lu and
Boutilier~\shortcite{budgetSocialChoice}, we know that the
$1$-best-\textsc{OWA-Winner} problem is $\np$-hard both for both
approval-based utilities and Borda-based utilities (in this case the
problem is equivalent to winner determination under appropriate
variants of Chamberlin--Courant voting rule; in effect, many results
regarding the complexity of this rule are applicable for this variant
of the
problem~\cite{fullyProportionalRepr,sko-fal-sli:c:multiwinner,yu-cha-elk:c:sp-tree,sko-fal:t:maxcover}).
A simple reduction shows that this result carries over to each family
of $k$-best OWAs and of $k$-med OWAs, where $k$ is a fixed positive
integer (note that for the case of approval-based utilities, these
results also follow through Thoerem~\ref{thm:almost-all-hard}).

\begin{proposition}\label{thm:k-best-k-med}
  For each fixed $k$, $k\text{-}\mathrm{best}$-\textsc{OWA-Winner} and
  $k\text{-}\mathrm{med}$-\textsc{OWA-Winner} are $\np$-complete, even if the
  utility profiles are restricted to be approval-based or Borda-based.
\end{proposition}

\begin{proof}
  Let $k$ be a fixed constant.  It is easy to see that
  $k$-best-\textsc{OWA-Winner} and $k$-med-\textsc{OWA-Winner}
  are both in $\np$. To show $\np$-hardness, we give reductions from
  $1$-best-\textsc{OWA-Winner} (either with approval-based utilities or with
  Borda-based utilities) to $k$-best-\textsc{OWA-Winner} and to
  $k$-med-\textsc{OWA-Winner} (with the same types of utilities).

  Let $I$ be an instance of $1$-best-\textsc{OWA-Winner} with $n$ agents, $m$
  items, and where we seek a winner set of size $K$. We form an
  instance $I'$ of $k$-best-\textsc{OWA-Winner} that is identical to $I$ except
  that: (1) We add $k-1$ special items $b_1, \ldots, b_{k-1}$
  such that under approval-based utilities each agent $i$ has utility
  $1$ for each item $b_j$, $1 \leq j \leq k-1$, and under
  Borda-based utilities each agent $i$ has utility $m+j-1$ for
  item $b_j$, $1 \leq j \leq k-1$. (2) We set the size of the
  desired winner set to be $K' = K+k-1$. It is easy to see that if
  there is an optimal solution $W'$ for $I'$ that achieves some
  utility $x$, then there is a solution $W''$ for $I'$ that uses all
  the $k-1$ items $b_1, \ldots, b_{k-1}$ and also achieves
  utility $x$. Further, the set $W''-\{b_1, \ldots, b_{k-1}\}$ is an
  optimal solution for $I$ and, for $I$, has utility $x -
  \sum_{i=1}^{n}\sum_{j=1}^{k-1}u_{i,b_j} = x -
  n\sum_{j=1}^{k-1}u_{1,b_j}$.

  Analogous argument shows that $1$-best-\textsc{OWA-Winner} reduces to
  $k$-med-\textsc{OWA-Winner} (also for approval-based and for
  Borda-based utilities).
\end{proof}

We leave the problem of generalizing the above two theorems to more
general classes of OWA vectors as a technical (but conceptually easy)
open problem.

\subsection{Rare Easy Cases}

While the \textsc{OWA-Winner} problem is in general $\np$-hard, there
are also some natural easy cases. For example, the problem is in $\p$
provided that we seek a winner set of a fixed size.  Naturally, in
practice the variant of the problem with fixed $K$ has only limited
applicability.

\begin{proposition}\label{pro:fixed-K}
  For each fixed constant $K$ (the size of the winner set), \textsc{OWA-Winner}
  is in $\p$.
\end{proposition}
\begin{proof}
  For a profile with $m$ items, there are only $O(m^K)$ sets of
  winners to try. We try them all and pick one that yields highest
  utility.
\end{proof}

Similarly, the problem is in $\p$ when the number of available items
is fixed (it follows by applying the above proposition; if the number
of items is fixed then so is $K$).  Throughout the rest of the paper
we focus on the $\alpha$-\textsc{OWA-Winner} variant of the problem,
where $K$ is given as part of the input and $\alpha$ represents a
family of OWAs, one for each value of $K$.

It is easy to note that for $K$-best OWA (that is, for the family of
constant OWAs $\alpha = (1, \ldots, 1)$) the problem is in $\p$.

\begin{proposition}\label{pro:K-in-p}
  $K\text{-}\mathrm{best}$-\textsc{OWA-Winner} is in $\p$.
\end{proposition}
\begin{proof}
  Let $I$ be an input instance with $m$ items and $n$ agents,
  where we seek a winner set of size $K$.  It suffices to compute for
  each item the total utility that all the agents would derive
  if this item were included in the winner set and return $K$
  items for which this value is highest.
\end{proof}

Indeed, if the agents' utilities are either approval-based or
Borda-based, $K$-best-\textsc{OWA-Winner} boils down to
(polynomial-time) winner determination for $K$-best approval rule and
for $K$-Borda rule~\cite{deb:j:k-borda}, respectively (see also the
work of Elkind et al.~\shortcite{elk-fal-sko-sli:c:multiwinner-rules}
for a general discussion of multiwinner rules).  However, in light of
this fact, Theorems~\ref{thm:k-1-best} and~\ref{thm:borda:k-1-best}
appear quite surprising.

Given the results in this section so far, we conjecture that the
family of constant OWAs, that is, the family of $K$-best OWAs, is the
only natural family for which $\alpha$-\textsc{OWA-Winner} is in
$\p$. We leave this conjecture as a natural follow-up
question.\footnote{It is tempting to conjecture that for {\em all}
  families of non-constant OWAs, not just the natural ones, the
  problem is $\np$-hard.  This, however, is not the case. Indeed, by
  following the arguments of the classic theorem of
  Ladner~\cite{lad:j:np-incomplete}, it is possible to show a
  polynomial-time computable family of OWAs such that
  $\alpha$-\textsc{OWA-Winner} is in $\np$, but is neither
  $\np$-complete nor in $\p$. (Intuitively put, such a family could
  consist of interspersed long fragments where the OWAs are either
  $K$-best or $1$-best.  The $K$-best fragments would prevent the
  problem from being $\np$-complete, while the $1$-best fragments
  would prevent it from being in $\p$.)}

\subsection{Integer Programming}

In spite of all the hardness results that we have seen so far, we
still might be in a position where it is necessary to obtain an exact
solution for a given $\alpha$-\textsc{OWA-Winner} instance, and where
the brute-force algorithm from Proposition~\ref{pro:fixed-K} is too
slow. In such a case, it might be possible to use an integer linear
programming (ILP) formulation of the problem, given below.  We believe
that this ILP formulation is interesting in its own right and, in
particular, that it is interesting future work to experimentally
assess the size of instances for which it yields solutions in
reasonable amount of time.

\begin{theorem}\label{thm:ilp}
  \textsc{OWA-Winner} reduces to computing a solution for the following integer
  linear program.
{\upshape
  \begin{align*}
    &\text{minimize }    \sum_{i=1}^n \sum_{j=1}^m \sum_{k=1}^K \alpha_k u_{i, a_j} x_{i, j, k}  \\
    &\text{subject to: }\\  
    & \text{(a)}: \sum_{i=1}^m x_i = K                                                                          \ \\
    & \text{(b)}: x_{i, j, k} \leq x_j                                                                             \ &  \!\!\!\!\!\!\!\!, i \in [n]; j, k \in [K] \\
    & \text{(c)}: \sum_{j=1}^m x_{i, j, k} = 1                                                                     \ &, i \in [n]; k \in [K] \\
    & \text{(d)}: \sum_{k=1}^K x_{i, j, k} = 1                        
                                             \ &, i \in [n]; j \in [m] \\
    & \text{(e)}: \sum_{j=1}^m u_{i, a_j} x_{i, j, k} \geq \sum_{j=1}^m u_{i,
    a_j} x_{i, j, (k+1)}                                   \ &, i \in [n]; k \in [K-1] \\
    & \text{(f)}: x_{i, j, k} \in \{0, 1\}                                                                         \ &   \!\!\!\!\!\!\!\!, i \in [n]; j, k \in [K] \\
    & \text{(g)}: x_{j} \in \{0, 1\} \ &, j \in [m]
  \end{align*}
}
\end{theorem}
\begin{proof}
  Consider an input instance with $n$ agents $N = [n]$ and $m$
  items $A = \{a_1, \ldots, a_m\}$, where we seek a winner set
  of size $K$, under OWA $\alpha = (\alpha_1, \ldots, \alpha_K)$. For
  each $i \in N$, $a_j \in A$, we write $u_{i,a_j}$ to denote the
  utility that agent $i$ derives from item $a_j$.

  We form an instance of ILP with the following variables: (1) For
  each $i \in N$, $j \in [m]$, and $k \in [K]$, there is an indicator
  variable $x_{i,j,k}$ (intuitively, we interpret $x_{i,j,k}=1$ to
  mean that for agent $i$, item $a_j$ is the $k$-th most
  preferred one among those selected for the solution).  (2) For each
  $j \in [m]$, there is an indicator variable $x_j$ (intuitively, we
  interpret $x_j=1$ to mean that $a_j$ is included in the
  solution). Given these variables (and assuming that we enforce their
  intuitive meaning), the goal of our ILP is to maximize the function
  $\sum_{i=1}^n \sum_{j=1}^m \sum_{k=1}^K \alpha_k u_{i, a_j} x_{i, j,
    k}$.  

  We require that our variables are indeed indicator
  variables and, thus, take values from the set $\{0,1\}$ only (constraints (f) and (g)).
  We requite that the variables of the form $x_{i,j,k}$ are internally
  consistent. (constraint (c) says that each agent ranks only one of
  the candidates from the solution as $k$-th best, constraint (d) say
  that there is no agent $i$ and item $a_j$ such that $i$ views
  $a_j$ as ranked on two different positions among the items
  from the solution.)
  Then, we require that variables of the form $x_{i,j,k}$ are
  consistent with those of the form $x_j$ (constraint (b)) and that exactly $K$
  items are selected for the solution (constraint (a)).

  Our final constraint, constraint (e), requires that variables
  $x_{i,j,k}$ indeed for each agent sort the items from the
  solution in the order of descending utility values. We mention that
  constraint (e) is necessary only for the case of OWAs $\alpha$ that
  are not-nonincreasing. For a nonincreasing $\alpha$, an optimal
  solution for our ILP already ensures the correct ``sorting''
  (otherwise our goal function would not be maximized).
\end{proof}

We should note that linear-programming formulations of OWA-based
optimization problems have appeared in the literature far before our
work; see, for example, the paper of Ogryczak and
\'Sliwinski~\cite{OgryczakS03}. Yet, we use the OWA operators in a
very different way and, thus, our approach is different. (In essence,
Ogryczak and \'Sliwi\'nski use an OWA operator to aggregate a number
of values, whereas we use a simple sum to aggregate the agents'
perceived utilities, but we compute these perceived utilities by
applying an OWA operator to each agent's individual, intrinsic
utilities.)

\section{Approximation: General Utilities and Approval Utilities}\label{sec:approx-general}
The \textsc{OWA-Winner} problem is particularly well-suited for
applications that involve recommendation systems (see, e.g., the work
of Lu and Boutilier~\shortcite{budgetSocialChoice} for a discussion of
$1$-best-OWA-Winner in this context). For recommendation systems it
often suffices to find good approximate solutions instead of perfect,
exact ones, especially if we only have estimates of agents' utilities.
It turns out that the quality of the approximate solutions that we can
produce for \textsc{OWA-Winner} very strongly depends on both the
properties of the particular family of OWAs used and on the nature of
agents' utilities.

First, we show that as long as our OWA is nonincreasing, a simple
greedy algorithm achieves $\left(1 - \frac{1}{e}\right)$ approximation ratio. This
result follows by showing that for a nonincreasing OWA $\alpha$, the
function $u^{\alpha}_\util$ (recall Definition~\ref{def:owa-winner}) is submodular and nondecreasing,
and by applying the famous result of Nemhauser et
al.~\shortcite{submodular}.

\newcommand{\utility}{\mathit{util}}
\newcommand{\gain}{\mathit{gain}}
\SetKwInput{KwNotation}{Notation}
\begin{algorithm}[t]
   \small
   \KwNotation{\\
   $\alpha_\ell \leftarrow$ input OWA operator $\alpha$, restricted to its top $\ell$ entries. \\
    }
   
   \hspace{1mm}\\
   \small
   \SetAlCapFnt{\small}
   $W \leftarrow \emptyset$\;
   \For{$\ell\leftarrow 1$ \KwTo $K$}{
      $\gain \leftarrow \{\}$\;
      \ForEach{$a \in A \setminus W$}{
          $\gain[a] \leftarrow u^{\alpha_\ell}_\util( W \cup \{a\} ) - u^{\alpha_{\ell-1}}_\util( W )$\;
      }
      $W \leftarrow W \cup \displaystyle\argmax_{a \in A \setminus W}( \gain[a] )$\;
   }
      \Return $W$\;

   \caption{\small The greedy algorithm for finding the utilitarian set of $K$ winners.}
   \label{alg:greedy}
\end{algorithm}

Recall that if $A$ is some set and $u$ is a function $u \colon 2^{A}
\rightarrow \realsplus$, then we say that: (1) $u$ is submodular if
for each $W$ and $W'$, $W \subseteq W' \subseteq A$, and each $a \in A
\setminus W'$ it holds that:
\[ 
  u(W \cup {a}) - u(W) \geq u(W' \cup {a}) - u(W'),
\] 
and (2) $u$ is nondecreasing if for each $W \subseteq A$ and
each $a \in A$ it holds that $u(W \cup \{a\}) \geq u(W)$.


\begin{lemma}\label{lemma:constOWA}
  Let $I$ be an instance of \textsc{OWA-Winner} with a nonincreasing
  OWA $\alpha$. The function $u^{\alpha}_{ut}$ is submodular and
  nondecreasing.
\end{lemma}
\begin{proof}
  Let $I$ be an instance of \textsc{OWA-Winner} with agent set $N =
  [n]$, item set $A = \{a_1, \ldots, a_m\}$, desired solution
  size $K$, and OWA $\alpha = \langle \alpha_1, \ldots,
  \alpha_K\rangle$. For each agent $i \in N$ and each item $a_j
  \in A$, $u_{i,a_j}$ is a nonnegative utility that $i$ derives from
  $a_j$.

  Since all the utilities and all the entries of the OWA vector are
  nonnegative, we note that $u^\alpha_\util$ is nondecreasing.  
  To show submodularity, we decompose $u^\alpha_\util$ as follows:
  \begin{align*}
  u^\alpha_\util(W) = \sum_{\ell=1}^{K-1}(\alpha_\ell -
  \alpha_{\ell+1})u^{\ell\text{-best-OWA}}_\util(W) + \alpha_Ku^{K\text{-best-OWA}}_\util(W)
  \end{align*}
  For each $W \subseteq A$, $i \in N$ and $\ell \in [m]$, let $\Top(W,
  i, \ell)$ be the set of those $\ell$ items from $W$ whose
  utility, from the point of view of agent $i$, is highest (we break
  ties in an arbitrary way).  Since nonnegative linear combinations of
  submodular functions are submodular, it suffices to prove that for
  each $i \in N$ and each $\ell \in [m]$, function $u_i^{\ell}(W) =
  \sum_{{w \in \Top(W, i, \ell)}} u_{i, w}$ is submodular.

  To show submodularity of $u_i^{\ell}$, consider two sets, $W$ and
  $W'$, $W \subseteq W' \subseteq A$, and some $a \in A \setminus W'$.
  We claim that:
  \begin{align}\label{ineq::submodularity}
    u_i^{\ell}(W \cup \{a\}) - u_i^{\ell}(W) \geq u_i^{\ell}(W' \cup
    \{a\}) - u_i^{\ell}(W') \textrm{.}
  \end{align}
  Let $u_{W}$ and $u_{W'}$ denote the utilities that the $i$-th agent
  has for the $\ell$-th best items from $W$ and $W'$,
  respectively (or $0$ if a given set has fewer than $\ell$
  elements). Of course, $u_{W'} \geq u_{W}$. Let $u_a$ denote $i$-th
  agent's utility for $a$. We consider two cases. If $u_a \leq u_{W}$,
  then both sides of \eqref{ineq::submodularity} have value
  0. Otherwise:
  \begin{align*}
    u_i^{\ell}(W' \cup \{a\}) - u_i^{\ell}(W') &= \max(u_a - u_{W'}, 0) \\
    u_i^{\ell}(W \cup \{a\}) - u_i^{\ell}(W) &= u_a - u_{W} \textrm{,}
  \end{align*}
  which proves \eqref{ineq::submodularity} and completes the
  proof.
\end{proof}

Based on the above result, we can easily show that
Algorithm~\ref{alg:greedy} is a polynomial time $(1 -
1/e)$-approximation for the \textsc{OWA-Winner} problem, for the case
of nonincreasing OWA vectors (see Theorem~\ref{thm:greedyAprox}
below). Algorithm~\ref{alg:greedy} is a natural incarnation of the
greedy algorithm of Nemhauser et al.~\shortcite{submodular}.  It
starts by setting the found-so-far solution $W$ to be empty.  Then, in
each iteration it extends $W$ by adding this item that causes the
greatest increase in the utility.

\begin{example}\label{exjl2}
  Let the items and agents be as in Example~\ref{exjl}. Let $K = 3$
  and consider OWA vector $\alpha = (2,1,0)$. Throughout the $K = 3$
  iterations, we obtain the following gain values (the contents of
  $W$ are given at the beginning of each iteration; below we also explain
  some of the computation):
  \[
  \setlength{\arraycolsep}{3pt}
  \begin{array}{c|cc|cccccc}
             & W & u_\util^{\alpha_{\ell-1}}(W) & \gain(a_1) & \gain(a_2) & \gain(a_3) & \gain(a_4) & \gain(a_5) & \gain(a_6)\\ 
    \hline
    \text{Iter.\ 1} & \emptyset & 0 &48 & 30 & 38 & 48 & 46 & 34\\  
    \text{Iter.\ 2} & \{a_1\} & 48 & - & 15 & 21 & 28 & 27  & 21 \\  
    \text{Iter.\ 3} & \{a_1,a_4\} & 76 & - & 2 & 7 & - & 8 & 5 
  \end{array}
  \]
  At the beginning of the first iteration $W = \emptyset$ and the
  algorithm simply computes the utility of each item separately, using
  OWA operator $\alpha_1 = \langle 2 \rangle$. For example,
  $u_\util^{\alpha_{1}}(\{a_1\}) = 2\cdot( 10 + 6 + 8 ) = 48$.  In the
  first iteration both $a_1$ and $a_4$ lead to the highest gain and,
  so, the algorithm is free to pick either of them. We assume it picks
  $a_1$.  In the second iteration, we have $W = \{a_1\}$ and, for
  example, the gain value for $a_4$ is computed as:
\[
u_\util^{\alpha_{2}}(\{a_1,a_4\}) - u_\util^{\alpha_{1}}(\{a_1\}) =
2\cdot(10+10+8) + (8+6+6) -48 = 76-48 = 28.
\]
It is the highest gain value and so the algorithm includes $a_4$ in
the solution. In the third iteration, item $a_5$ has the highest gain
and so the algorithm includes it in $W$. Finally, the algorithm
outputs $W = \{a_1,a_4, a_5\}$.
\end{example}

\begin{theorem}\label{thm:greedyAprox}
  For a nonincreasing OWA $\owa$, Algorithm~\ref{alg:greedy} is a
  polynomial time $(1 - 1/e)$-approximation algorithm for the problem
  of finding the utilitarian set of $K$ winners.
\end{theorem}
\begin{proof}
  The thesis follows from the results of Nemhauser et
  al.~\shortcite{submodular} on approximating nondecreasing submodular
  functions.
\end{proof}

Algorithm~\ref{alg:greedy} has interesting interpretation in the
context of voting systems. This greedy algorithm can be viewed not
only as an approximation algorithm, but also as a new iterative voting
rule. Indeed, many popular voting rules are defined as iterative
(greedy) algorithms. Such rules are not only polynomially solvable,
but also are easier to understand for the society. Further,
Caragiannis et al.~\cite{car-kak-kar-pro:j:dodgson-acceptable} and,
later, Elkind~et~al.~\cite{elk-fal-sko-sli:c:multiwinner-rules},
advocate viewing approximation algorithms for computationally hard
voting rules as new election systems, and study their axiomatic
properties (often showing that they are better than those of the
original rules). 

Here we give another interesting observation. It turns out that the
algorithm from Theorem~\ref{thm:greedyAprox}, when applied to the case
of approval-based utilities and the harmonic OWA, is simply the winner
determination procedure for the \emph{Sequential Proportional Approval
  Voting} rule~\cite{RePEc:pra:mprapa:22709} (developed by the Danish
astronomer and mathematician Thorvald~N.~Thiele, and used for a short
period in Sweden during early 1900's). That is, the Sequential
Proportional Approval Voting rule is simply an approximation of the
PAV rule (the Proportional Approval Voting rule). We believe that this
observation gives another evidence that approximation algorithms for
computationally hard voting rules can indeed be viewed as new
full-fledged voting rules. (We point readers interested in
approval-based multiwinner voting rules to the overview of
Kilgour~\cite{pavVoting} and to the works of Aziz et
al.~\cite{azi-gas-gud-mac-mat-wal:c:approval-multiwinner,azi-bri-con-elk-fre-wal:c:justified-representation},
Elkind and Lackner~\cite{elk-lac:c:dichotomous-prefs}, and Skowron and
Faliszewski~\cite{sko-fal:t:maxcover}).

Is a $(1-\frac{1}{e})$-approximation algorithm a good result? After
all, $1-\frac{1}{e} \approx 0.63$ and so the algorithm guarantees only
about 63\% of the maximum possible satisfaction for the
agents. Irrespective if one views it as sufficient or not, this is the
best possible approximation ratio of a polynomial-time algorithm for
(unrestricted) \textsc{OWA-Winner} with a nonincreasing OWA. The
reason is that $1$-best-OWA-Winner with approval-based utilities is,
in essence, another name for the \textsc{MaxCover} problem, and if $\p
\neq \np$, then $(1-\frac{1}{e})$ is approximation upper bound for
\textsc{MaxCover}~\cite{Feige:1998:TLN:285055.285059}. We omit the
exact details of the connection between \textsc{MaxCover} and
$1$-best-OWA-Winner and instead we point the readers to the work of
Skowron and Faliszewski~\shortcite{sko-fal:t:maxcover} who discuss
this point in detail (we mention that they refer to what we call
$1$-best-OWA-Winner as winner determination for Chamberlin--Courant's
voting rule).

For OWAs that are not nonincreasing, it seems that we cannot even hope
for a $(1-\frac{1}{e})$-approximation algorithm. There are two
arguments to support this belief. First, such OWAs yield utility
functions that are not necessarily submodular and, so, it is
impossible to apply the result of Nemhauser et al.~\cite{submodular}.
As an example, we show that $2$-med-OWA yields a utility function that
is not submodular.

\begin{example}\label{example:non-sub}
  Let us consider a single agent, two sets of items $W = \{c,
  d\}$ and $W' = \{b, c, d\}$ (of course $W \subset W'$), and
  $2\text{-}\mathrm{med}$-OWA $\alpha$. The utilities of the agent over the
  items $a$, $b$, $c$, and $d$ are equal to 10, 9, 2, and 1,
  respectively. We get:
\begin{align*}
& u_\util^{\alpha}(W \cup \{a\}) - u_\util^{\alpha}(W) = 2 - 1 = 1, 
& u_\util^{\alpha}(W' \cup \{a\}) - u_\util^{\alpha}(W') = 9 - 2 = 7 \textrm{.} 
\end{align*}
That is, $u_\util^{\alpha}$ is not submodular. Indeed, this example
works even for approval-based utilities: it suffices to set the
utilities for $a$ and $b$ to be $1$, and for $c$ and $d$ to be $0$.
\end{example}

Second, it is quite plausible that there are no constant-factor
approximation algorithms for many not-nonincreasing OWAs.  As an
example, let us consider the case of families of OWAs with the
following structure: their first entries are zeros followed by some
nonzero entry at a sufficiently early position. If there were a good
approximation algorithm for winner determination under such OWAs, then
there would be a good approximation algorithm for the
\textsc{Densest-K-Subgraph} problem, which seems unlikely.

\begin{definition}
  In a \textsc{Densest-k-Subgraph} problem we are given an undirected
  graph $G = (V, E)$ and a positive integer $K$. We ask for a subgraph
  $S$ with $K$ vertices with the maximal number of edges.
\end{definition}

\begin{theorem}\label{thm:dks}
  Fix some integer $\ell$, $\ell \geq 2$.  Let $\alpha$ be a family of
  OWAs such that each OWA in the family (for at least $\ell$ numbers)
  has $0$s on positions $1$ through $\ell-1$, and has a nonzero value
  on the $\ell$'th position. If there is a polynomial-time
  $x(n)$-approximation algorithm for $\alpha$-\textsc{OWA-Winner} then
  there is a polynomial-time $x(n)$-approximation algorithm for the
  \textsc{Densest-k-Subgraph} problem.
\end{theorem}
We should mention that Theorem~\ref{thm:dks} holds for a somewhat more
general class of OWAs than stated explicitly. The proof relies on the
fact that the first entry of the OWA is zero and that after the first
non-zero entry of the OWA there are still $K-1$ positions, where $K$
is the parameter from the input \textsc{Densest-K-Subset}
instance.\medskip

\begin{proof}[Proof of Theorem~\ref{thm:dks}]
  Let $I$ be an instance of the \textsc{Densest-K-Subgraph} problem
  with graph $G = (V, E)$ and positive integer $K$.  From $I$ we
  construct an instance $I'$ of $\alpha$-\textsc{OWA-Winner}, where
  the set of agents $N$ is $E$, the set of items is $A = V \cup
  \{d_1, \ldots, d_{\ell-2}\}$ (or $V$ if $\ell = 2$), and we seek a
  winner set of size $K+\ell-2$. Agents utilities are set as follows:
  For each agent $e$ and each item $d_j$, $1 \leq j \leq
  \ell-2$, the utility of $e$ for $d_j$ is $1$. If $e$ is an edge in
  $G$ than connects vertices $u$ and $v$ then agent $e$'s utility for
  $u$ and $v$ is $1$ and is $0$ for the remaining items from
  $V$.

  It is easy too see that the items $d_1, \ldots, d_{\ell-2}$
  all belong to every optimal solution for $I'$. It is also easy to
  see that in each optimal solution the utility of each agent $e$ is
  nonzero (and exactly equal to $\alpha_{\ell}$, the $\ell$-th entry
  of the OWA $\alpha$ used) if and only if both items
  corresponding to the vertices connected by $e$ are included in the
  solution. Thus the total utility of every optimal solution for $I'$
  is equal to $\alpha_{\ell}$ times the number of edges that connect
  any two vertices corresponding to the items from the
  solution.

  Let $\mathcal{A}$ be a polynomial-time $x(n)$-approximation algorithm
  for $\alpha$-\textsc{OWA-Winner}. If $\mathcal{A}$,
  returns a solution $S$ for $I'$ with none-zero utility, then the items $d_1, \ldots, d_{\ell-2}$
  all belong to $S$. Let us take the vertices corresponding to the
  items $S \setminus \{d_1, \ldots, d_{\ell-2}\}$. The number of the edges connecting
  these vertices is equal to the total utility of $S$ divided by $\alpha_{\ell}$.
  Thus, from $x(n)$-approximation solution for $I'$ we can extract an
  $x(n)$-approximation solution for $I$. This completes the proof.
\end{proof}

It seems that the \textsc{Densest-k-Subgraph} is not easy to
approximate. Khot~\shortcite{Khot_rulingout} ruled out the existence
of a PTAS for the problem under standard complexity-theoretic
assumptions, Bhaskara~et~al.~\shortcite{conf/soda/BhaskaraCVGZ12}
showed polynomial integrality gap, Raghavendra and
Steurer~\shortcite{conf/stoc/RaghavendraS10} and
Alon~et~al.~\shortcite{techreport/inaaproxDkS} proved that there is no
polynomial-time constant approximation under non-standard
assumptions. Finally, the best approximation algorithm for the problem
that we know of, due to
Bhaskara~et~al.~\shortcite{conf/stoc/BhaskaraCCFV10}, has
approximation ratio $O(n^{1/4 + \epsilon})$, where $n$ is the number
of vertices in the input graph.

As a further evidence that OWAs that are not nonincreasing are
particularly hard to deal with from the point of view of approximation
algorithms, we show that for an extreme example of an OWA family,
i.e., for the $K$-med OWAs, there is a very strong
hardness-of-approximation result. We start from the following graph
problem.

\begin{definition}
  In the \textsc{Maximum Edge Biclique Problem} (\textsc{MEBP}) we are
  given a balanced bipartite graph $(U \cup V, E)$ where $U \cup V$ is
  the set of vertices ($\|U\| = \|V\|$) and $E$ is the set of edges
  (there are edges only between the vertices from $U$ and $V$). We ask
  for a biclique (i.e., a subgraph $S$, such that every vertex from $U
  \cap S$ is connected with every vertex from $V \cap S$) with as many
  edges as possible.
\end{definition}

According to Feige~and~Kogan~\shortcite{Feige04hardnessof}, there exists a
constant $c$ such that there is no polynomial $(2^{c\sqrt{\lg
    n}}/n)$-approximation algorithm for \textsc{MEBP} unless for some
$\epsilon$ we have $\textsc{3-SAT} \in \mathrm{DTIME}(2^{n^{3/4 +
    \epsilon}})$. Currently it seems unlikely that such an algorithm
for \textsc{3-SAT} exists.
For our argument it is more convenient to
define and use the following variant of \textsc{MEBP}.

\begin{definition}
  In \textsc{MEBP-V} we are given the same input as in \textsc{MEBP}
  and a positive integer $K$. We ask for a biclique $S$ such that $\|S
  \cap V\| = K$ and $S$ contains as many edges as possible.
\end{definition}

\begin{lemma}\label{lemma:MEBP-V}
  There exists a constant $c$ such that there is no polynomial-time
  $(2^{c\sqrt{\lg n}}/n)$-approximation algorithm for \textsc{MEBP-V}
  unless for some $\epsilon$ we have $\textsc{3-SAT} \in
  \mathrm{DTIME}(2^{n^{3/4 + \epsilon}})$.
\end{lemma}
\begin{proof}
  For the sake of contradiction, let us assume that there exists a
  constant $c$ and a polynomial-time $(2^{c\sqrt{\lg
      n}}/n)$-approximation algorithm $\mathcal{A}$ for
  \textsc{MEBP-V}. By running $\mathcal{A}$ for every value of $K$
  ranging from $1$ to $\|V\|$, we obtain a polynomial-time
  $(2^{c\sqrt{\lg n}}/n)$-approximation algorithm for
  \textsc{MEBP}. This stays in contradiction with the result of
  Feige~and~Kogan~\shortcite{Feige04hardnessof}.
\end{proof}

\begin{theorem}\label{thm:mebp-bounded}
  There exists a constant $c$ such that there is no polynomial-time
  $(2^{c\sqrt{\lg n}}/n)$-approximation algorithm for
  $K\text{-}\mathrm{med}$-\textsc{OWA-Winner} unless for some $\epsilon$ we
  have $\textsc{3-SAT} \in \mathit{DTIME}(2^{n^{3/4 + \epsilon}})$.
\end{theorem}
\begin{proof}
  Let us assume that there is a constant $c$ and a polynomial-time
  $(2^{c\sqrt{\lg n}}/n)$-approximation algorithm $\mathcal{A}$ for
  $K$-med-\textsc{OWA-Winner}. We will show that we can use
  $\mathcal{A}$ to solve instances of \textsc{MEBP-V} with the same
  approximation ratio. By Lemma~\ref{lemma:MEBP-V}, this will prove
  our theorem.

  Let $I$ be an instance of \textsc{MEBP-V} with bipartite graph $G =
  (U \cup V, E)$ and positive integer $K$.  From $I$ we construct an
  instance $I'$ of $K$-med-\textsc{OWA-Winner} in the
  following way. We let the set of agents $N$ be $U$, the set of
  items $A$ be $V$, and we seek a winner set of size $K$. The
  utility of agent $u$ from item $v$ is equal to $1$ if and
  only if $u$ and $v$ are connected in $G$; otherwise it is $0$. Now
  we note that there is a one-to-one correspondence between the solutions
  for $I$ and for $I'$. Let $S$ be a solution for $I$ with $x$ edges:
  $S \cap V$ is also a solution for $I'$ with the utility at least equal to $x/K$. 
  Let $S$ be a solution for $I'$ with the utility $x$. All the agents from
  $U$ with non-zero utilities, together with $S$, form a biclique with $Kx$ edges.
  Thus, from the $(2^{c\sqrt{\lg n}}/n)$-approximation solution for $I'$ we
  can extract a $(2^{c\sqrt{\lg n}}/n)$-approximation solution for $I$.
  This completes the proof.
\end{proof}

As a corollary of the above proof, we also have that
$\mathrm{Hurwicz[}\lambda\mathrm{]}$-\textsc{OWA-Winner} is $\np$-hard
(through an almost identical proof, but with a certain dummy candidate
added, that gets utility $1$ from everyone, and with the size of the
winner set extended by $1$).

\begin{corollary}\label{cor:hurwicz:np-hard}
$\mathrm{Hurwicz[}\lambda\mathrm{]}$-\textsc{OWA-Winner} is $\np$-hard
\end{corollary}

The reader may wonder why for the case of
$\mathrm{Hurwicz[}\lambda\mathrm{]}$ OWA we only obtain $\np$-hardness
and not inapproximability. The reason is that due to the added dummy
candidate it is easy to find a winner set with nonnegligible
utility. In fact, this is a general property of the
$\mathrm{Hurwicz[}\lambda\mathrm{]}$ OWA and we show an approximation
algorithm for it with a constant approximation ratio. This shows that
even for OWAs that are not nonicreasing it is sometimes possible to
find positive approximation results (though later we will argue that
this approximation is not fully satisfying).

\begin{proposition}\label{pro:hurwicz}
  Let $\calA$ be a $\beta$-approximation algorithm for
  $1\text{-}\mathrm{best}$-\textsc{OWA-Winner}.
  $\calA$ is a $\lambda\cdot\beta$-approximation algorithm for
  $\mathrm{Hurwicz[}\lambda\mathrm{]}$-\textsc{OWA-Winner}.
\end{proposition}
\begin{proof}
  Let us consider some instance $I^H$ of
  $\mathrm{Hurwicz[}\lambda\mathrm{]}$-\textsc{OWA-Winner}, where the
  goal is to pick a set of $K$ items. We construct an instance $I^1$
  that is identical to $I^H$, but for the $1$-best-OWA, and we run
  algorithm $\calA$ on $I^1$. The algorithm outputs some set $W =
  \{w_1, \ldots, w_K\}$ (a $\beta$-approximate solution for $I^1$).
  We claim that $W$ is a $\lambda\beta$-approximate solution for
  $I^H$.

  Let $W^H = \{w_1^H, \ldots, w_K^H\}$ be an optimal solution for
  $I^H$ and let $W^1 = \{w_1^1, \ldots, w_K^1\}$ be an optimal
  solution for $I^1$.  We first note that the following holds (recall
  the $\vec x^{\downarrow}$ notation for sorted sequences):
  \begin{align*}
  u_{ut}^{\mathrm{Hurwicz[}\lambda\mathrm{]}}(W^H) = \sum_{i=1}^n\left( \lambda u_{i, w^H_1}^{\downarrow} + (1-\lambda) u_{i, w^H_K}^{\downarrow}\right) \leq
  \sum_{i=1}^nu_{i, w^H_1}^{\downarrow} \leq
  \sum_{i=1}^nu_{i, w^1_1}^{\downarrow} =
  u_{ut}^{1\text{-}\mathrm{best}}(W^1).
  \end{align*}
  In effect, we have that $u_{ut}^{1\text{-}\mathrm{best}}(W^1) \geq
  u_{ut}^{\mathrm{Hurwicz[}\lambda\mathrm{]}}(W^H)$.  Now, it is easy
  to verify that for $W$ (or, in fact, for any set of $K$ items) it
  holds that:
  \begin{align*}
    u_{ut}^{\mathrm{Hurwicz[}\lambda\mathrm{]}}(W) =
    \sum_{i=1}^n\left( \lambda u_{i, w_1}^{\downarrow} + (1-\lambda)
      u_{i, w_K}^{\downarrow}\right)  \geq \lambda \sum_{i=1}^n u_{i,
      w_1}^{\downarrow} = \lambda u_{ut}^{1\text{-}\mathrm{best}}(W)
    \text{.}
  \end{align*}
  Finally, combining these two inequalities and the fact that $W$ is a $\beta$-approximate
  solution for $1$-best\textsc{OWA-Winner}, we get:
  \[
     u_{ut}^{\mathrm{Hurwicz[}\lambda\mathrm{]}}(W) \geq \lambda u_{ut}^{1\text{-}\mathrm{best}}(W) 
     \geq \lambda\beta u_{ut}^{1\text{-}\mathrm{best}}(W^1) \geq \lambda\beta u_{ut}^{\mathrm{Hurwicz[}\lambda\mathrm{]}}(W^H).
  \]
  This completes the proof.
\end{proof}

By using Algorithm~\ref{alg:greedy} in the general case, and the PTAS
of Skowron et al.~\shortcite{sko-fal-sli:c:multiwinner} for
$1$-best-\textsc{OWA-Winner} with Borda-based utilities, we get the
following corollary.

\begin{corollary}\label{cor:hurwicz}
  (1) There is an algorithm that for
  $\mathrm{Hurwicz[}\lambda\mathrm{]}$-\textsc{OWA-Winner} with no
  restrictions on the utility functions achieves approximation ratio
  $\lambda(1-\frac{1}{e})$.  (2) For each positive $\epsilon$, there
  is an algorithm that for
  $\mathrm{Hurwicz[}\lambda\mathrm{]}$-\textsc{OWA-Winner} for the
  case of Borda-based utilities achieves approximation ration
  $\lambda(1-\epsilon)$.
\end{corollary}

Nonetheless, Corollary~\ref{cor:hurwicz} has a bitter-sweet taste. In
essence, it says that instead of using Hurwicz[$\lambda$] OWAs, we
might as well use $1$-best OWAs. If one wanted to use
Hurwicz[$\lambda$] OWAs for some important reason, then our
approximation result would not be sufficient. Yet, from a different
perspective, one could interpret Corollary~\ref{cor:hurwicz} as
suggesting that such an important reason is unlikely to exist (for
large values of $\lambda$).

Nonetheless, the idea of using a simpler OWA instead of a more complex
one can lead to quite intriguing results. Based on this approach,
below we show a PTAS for \textsc{OWA-Winner} for a family OWAs that
are similar to $K$-best OWAs (this restriction is necessary to defeat
the relation with the \textsc{MaxCover} problem which precludes
arbitrarily good approximation algorithms).

\begin{theorem}\label{thm:approx}
  Consider a nonincreasing OWA $\alpha$,
  $\alpha = \langle \alpha_1, \ldots, \alpha_K\rangle$. Let $I$ be an
  instance for $\alpha$-\textsc{OWA-Winner} (where we seek a winner
  set of size $K$). An optimal solution for the same instance but with
  $K\text{-}\mathrm{best}$-OWA is a $(\sum_{i=1}^K\alpha_i)/(K\alpha_1)$-approximate solution
  for $I$.
\end{theorem}
\begin{proof}
  Let $I$ be the instance of $\alpha$-\textsc{OWA-Winner} described in
  the statement of the theorem, let $W^{*}$ be one of its optimal
  solution, and let $W$ be an optimal solution for the same
  instance, but with the $K$-best-OWA. Note that $W$ is also an
  optimal solution for the $K$-number constant OWA $\beta = \langle \alpha_1, \ldots, \alpha_1 \rangle$. We claim
  that the following inequalities hold ($u^\alpha_\util$ is defined
  with respect to the instance $I$ and $u^\beta_\util$ is defined with
  respect to instance $I$ with $\beta$ as the OWA):
  \begin{align*}
    u_{ut}^{\owa}(W) \geq \frac{\sum_{i=1}^K\alpha_i}{K\alpha_1} u_{ut}^{\owab}(W)
    \geq \frac{\sum_{i=1}^K\alpha_i}{K\alpha_1} u_{ut}^{\owab}(W^{*}) 
    \geq \frac{\sum_{i=1}^K\alpha_i}{K\alpha_1} u_{ut}^{\owa}(W^{*}) \textrm{,}
  \end{align*}
  The second inequality holds because $W$ is an optimal solution
  for $I$ with OWA $\beta$. To see why the first and the third
  inequalities hold, let us focus on some agent $i$. The third inequality
  is simpler and so we prove it first.

  Let $u^{*}_1, \ldots, u^{*}_k$ be the utilities, in the
  nonincreasing order, that agent $i$ has for the items in
  $W^{*}$. Thus the utility that $i$ gets from $W^{*}$ under $\alpha$
  is $\sum_{i=1}^{K} \alpha_iu^{*}_i$. Since for each $i$, $1 \leq i
  \leq K$, we have $\alpha_i \leq \alpha_1$, $i$'s utility under
  $\alpha$ is less or equal to $i$'s utility under $\beta$,
  $\sum_{i=1}^{K} \alpha_1u^{*}_i$.

  We now prove the first inequality. Let $u_1, \ldots, u_K$ be the
  utilities, in the nonincreasing order, that agent $i$ has for the
  items in $W$. Our goal is to show that:
  \begin{align*}
    \alpha_1u_1 + \cdots + \alpha_Ku_K \geq
    \frac{\sum_{i=1}^K\alpha_i}{K\alpha_1}\alpha_1u_1 + \cdots + \frac{\sum_{i=1}^K\alpha_i}{K\alpha_1}\alpha_1u_K
    = \frac{\sum_{i=1}^K\alpha_i}{K}u_1 + \cdots + \frac{\sum_{i=1}^K\alpha_i}{K}u_K.
  \end{align*}
  This inequality is equivalent to
  \begin{align*}
    K\alpha_1u_1 + \cdots + K\alpha_Ku_K \geq \sum_{i=1}^K\alpha_iu_1
    + \cdots + \sum_{i=1}^K\alpha_iu_K,
  \end{align*}
  which itself is equivalent to
  \begin{align*}\textstyle
    u_1(K\alpha_1-\sum_{i=1}^K\alpha_i) + \cdots + u_K(K\alpha_K-\sum_{i=1}^K\alpha_i) \geq 0.
  \end{align*}
  We can rewrite the left-hand side of this inequality as:
  \begin{align*}
    &\textstyle (u_1-u_2)(K\alpha_1-\sum_{i=1}^K\alpha_i) +(u_2-u_3)(K\alpha_1+K\alpha_2 -2\sum_{i=1}^K\alpha_i) +\cdots+ \\
    &\textstyle +(u_{K-1}-u_K)(\sum_{j=1}^{K-1}K\alpha_j -(K-1)\sum_{i=1}^K\alpha_i) +u_K(\sum_{j=1}^{K}K\alpha_j -K\sum_{i=1}^K\alpha_i).
  \end{align*}
  We claim that each summand in this expression is nonnegative.  Since
  $u_1, \ldots, u_K$ is a nonincreasing sequence of nonnegative
  utilities, we have that for each $j$, $1 \leq j \leq K-1$,
  $u_j-u_{j+1}$ is nonnegative, and so is $u_K$. Now fix some $t$, $1
  \leq t \leq K$. We have:
  \begin{align*}\textstyle
    \sum_{j=1}^tK\alpha_j - t\sum_{i=1}^K \alpha_i &\textstyle= \sum_{j=1}^t(K-t)\alpha_j - t\sum_{i=t+1}^K\alpha_i\\
    &\geq t(K-t)\alpha_t - t\sum_{i=t+1}^K\alpha_i 
    \textstyle \geq t(K-t)\alpha_t - t(K-t)\alpha_t = 0
  \end{align*}
  This completes the proof.
\end{proof}

As a consequence of this theorem, we immediately get the following
result.
\begin{theorem}\label{thm:k-1-best-ptas}
  Let $f: \naturals \rightarrow \naturals$ be a function computable in
  polynomial-time with respect to the value of its argument, such that
  $f(K)$ is $o(K)$. There is a PTAS for $(K-f(K))\text{-}\mathrm{best}$-\textsc{OWA-Winner}.
\end{theorem}
\begin{proof}
  Let us fix some $\epsilon$, $0 < \epsilon <1$. We give a polynomial
  time $\epsilon$-approximation algorithm for
  $(K-f(K))$-best-\textsc{OWA-Winner}. Since $f(K)$ is $o(K)$, there
  is some value $K_\epsilon$ such that for each $K \geq K_\epsilon$ it
  holds that $\frac{K-f(K)}{K} \geq \epsilon$. If for our input
  instance we are to find a winner set of size $K$, $K \geq
  K_\epsilon$, then we simply run the polynomial-time algorithm for
  $K$-best-OWA.  Otherwise, we seek a winner set of size at most
  $K_\epsilon$ and we try all subsets of items of size
  $K$. Since, in this case, $K$ is bounded by a constant, our
  algorithm runs in polynomial time.
\end{proof}

While Theorem~\ref{thm:k-1-best-ptas} suffers from the same criticism
as Corollary~\ref{cor:hurwicz}, it is still a very interesting result,
especially when contrasted with
Theorem~\ref{thm:dks}. Theorem~\ref{thm:k-1-best-ptas} says that there
is a PTAS for $\alpha$-\textsc{OWA-Winner} for OWA family $\langle 1,
\ldots, 1, 0\rangle$, whereas Theorem~\ref{thm:dks} suggests that it
is unlikely that there is a constant-factor approximation algorithm
for $\alpha$-\textsc{OWA-Winner} with OWA family $\langle 0,1, \ldots,
1 \rangle$. Even though these two OWA families seem very similar, the
fact that one is nonincreasing and the other one is not makes a huge
difference in terms of approximability of \textsc{OWA-Winner}.

\section{Approximation: Non-Finicky Utilities}\label{sec:approx-borda}

One of the greatest sources of hardness of the \textsc{OWA-Winner}
problem, that we rely on in our proofs, is that the agents may have
very high utilities for some very small subsets of items, and very low
utilities for the remaining ones (consider, e.g., approval-based
utilities where each agent approves of relatively few items). In such
cases, intuitively, either we find a perfect solution or some of the
agents have to be very badly off. On the other hand, for Borda-based
utilities when some agent does not get his or her top items, it is
still possible to provide the agent with not-much-worse ones; the
utilities decrease linearly. Indeed, Skowron et
al.~\shortcite{sko-fal-sli:c:multiwinner} used this observation to
give a PTAS for the Chamberlin--Courant rule. Here we give a strong
generalization of their result that applies to non-finicky utilities
and OWA families that include, for each fixed$k$, $k$-median,
$k$-best, and geometric progression OWAs.

We focus on the case of OWA vectors where only some constant number
$\ell$ of top positions are nonzero, and on
$(\beta,\gamma)$-non-finicky utilities ($\beta, \gamma \in [0,1]$). In
this case, Algorithm~\ref{alg:greedy2} (a generalization of an
algorithm of Skowron~et~al.~\shortcite{sko-fal-sli:c:multiwinner})
achieves a good approximation ratio.  The idea behind the algorithm is
as follows: To pick $K$ items, it proceeds in $K$ iterations and in
each iteration it introduces one new item into the winner set. For
each agent it considers the top $x = \gamma m$ items with the highest
utilities and in a given iteration it picks an item $a$ that maximizes
the number of agents that (1) rank $a$ among items with the highest
$x$ utilities, and (2) still have ``free slots'' (an agent has a free
slot if among the so-far-selected winners, fewer than $\ell$ have
utilities among the $x$ highest ones for this agent). 
Before we prove
that our algorithm works well, let us consider the following example.

\newcommand{\rank}{\mathrm{rank}}

\SetKwInput{KwNotation}{Notation}
\begin{algorithm}[t]
   \small
   \SetAlCapFnt{\small}
    \small
   \SetAlCapFnt{\small}
   \KwNotation{\\
   $\Phi \leftarrow$ a map giving the number of free slots per agent; at first, for each agent $i$ we have $\Phi[i] = \ell$.\\
   $\rank(j, a) = \|\{b \in A \colon u_{j, b} > u_{j, a}\}\|$ gives the rank of item $a$ according to agent $i$.\\}
   \hspace{1mm}\\
   $x \leftarrow \gamma m$\;
   $S \leftarrow \emptyset$\;
   \For{$i\leftarrow 1$ \KwTo $K$}{
      $a \leftarrow \mathrm{argmax}_{a \in A \setminus S} \|\{j \mid \Phi(j) > 0 \land \rank(j, a) < x\}\|$\;
      \ForEach{$j \in \{j \mid \Phi(j) > 0\}$}{
         \If{$\rank(j, a) < x$}
         {
           $\Phi[j] \leftarrow \Phi[j] - 1$\;
         }
      }
      $S \leftarrow S \cup \{a\}$\; 
   }
   \Return{$S$}

   \caption{An algorithm for nonincreasing OWAs where at most first
     $\ell$ entries are nonzero, for the case of $(\beta, \gamma)$-non-finicky utilities.}\label{alg:greedy2}
\end{algorithm}

\begin{example}\label{expf}
  Let the items and agents be the same as in Example~\ref{exjl} (just
  as in Example~\ref{exjl2} for Algorithm~\ref{alg:greedy}). Let $K =
  3$ and let the OWA vector be $\alpha = (2,1,0)$. We have $\ell = 2$
  nonzero entries in $\alpha$. We treat the agents utilities as
  $(0.8,0.5)$-non-finicky ones.  Before we execute the algorithm, it
  is convenient to compute the $\rank$ function:\footnote{Note that
    here the best rank is $0$ and not $1$ (using rank $1$ for the top
    item is the more common approach). This simplifies our technical
    discussion.}
  \[
  \begin{array}{c|cccccc}
                     & a_1 & a_2 & a_3 & a_4 & a_5 & a_6\\ 
    \hline
    \rank(1, \cdot ) & 0 & 0 & 2 & 3 & 4 & 5 \\
    \rank(2, \cdot )  & 3 & 4 & 5 & 0 & 2 & 0 \\
    \rank(3, \cdot ) & 2 & 5 & 0 & 4 & 0 & 3 
  \end{array}
  \]
  Now we can start to execute the algorithm. We have $x = \gamma m =
  3$ and initially each agent has two free slots.  In the first
  iteration, the algorithm can pick either $a_1$, $a_3$, or $a_5$,
  because for each of them there are two agents for whom their rank is
  below $3$, while for each other item there is only one agent that
  ranks it below $3$. Let us assume that the algorithm picks $a_1$
  (see the table below for information regarding the slots of the
  agents after each iteration). In the second iteration all the agents
  still have  free slots so the algorithm can pick either $a_3$ or $a_5$.
  Let us assume it picks $a_3$. In effect, Agents $1$ and $3$ no longer
  have free slots and in the final iteration the algorithm picks one
  of the items to which Agent~$2$ assigns rank lower than $3$, i.e.,
  one of $a_4$, $a_5$, and $a_6$. Let us assume it picks $a_4$.
  Below we show the contents of agents' slots after executing each
  iteration.
  \[
  \begin{array}{c|cc|cc|cc}
                     & \multicolumn{2}{c|}{\mathtt{Agent\ 1}}  & \multicolumn{2}{c|}{\mathtt{Agent\ 2}} & \multicolumn{2}{c}{\mathtt{Agent\ 3}}\\ 
   & \text{Slot 1} & \text{Slot 2} & \text{Slot 1} & \text{Slot 2} & \text{Slot 1} & \text{Slot 2} \\
    \hline
    \text{After iteration 1} & a_1 & - & - & - & a_1 & - \\
    \text{After iteration 2}  & a_1 & a_3 & - & - & a_1 & a_3 \\
    \text{After iteration 3} & a_1 & a_3 & a_4 & - & a_1 & a_3 
  \end{array}
  \]
  The algorithm outputs set $S = \{a_1, a_3, a_4\}$. It is interesting
  that this is a different set than the one returned by
  Algorithm~\ref{alg:greedy} (see Example~\ref{exjl2}), which returned
  set $W = \{a_1, a_4, a_5\}$. This latter set is slightly
  better than the one output by Algorithm~\ref{alg:greedy2}; it
  achieves aggregated utility $84$ as opposed to $83$.
\end{example}

\begin{theorem}\label{thm:borda:nonincreasing}
  Fix a positive integer $\ell$ and let $\owa$ be a nonincreasing OWA
  where at most first $\ell$ entries are nonzero. 
  If the agents have $(\beta,  \gamma)$-non-finicky utilities,
  with $\gamma m \geq  \ell$, then 
  Algorithm~\ref{alg:greedy2} is a polynomial-time $\beta(1 - \exp(-\frac{\gamma
    K}{\ell}))$-approximation algorithm for $\alpha$-\textsc{OWA-Winner}.
\end{theorem}
\begin{proof}
  Consider an instance $I$ of $\alpha$-\textsc{OWA-Winner}, with $n$
  agents, $m$ items, and where we seek a winner set of size $K$.  Let
  $x = \gamma m$.  We use an OWA where an agent's total utility from a
  winner set $W$ depends on this agent's utilities for his or her top
  $\ell$ items from $W$. We introduce the notion of each agent's free
  slots as follows. Initially, each agent has $\ell$ free slots.
  Whenever an agent $j$ has a free slot and the algorithm selects an
  item $a$ such that for agent $j$ item $a$ is among $x$ items with
  highest utilities, we say that $a$ starts occupying one free slot of
  $j$. After such an item is selected, $j$ has one free slot less.

  Let $n_i$ denote the total number of free slots of all the agents
  after the $i$-th iteration of the algorithm. Naturally, we have $n_0
  = \ell n$. We show by induction that $n_i \leq \ell n \left(1 -
    \frac{x}{\ell m}\right)^i$. Indeed, the inequality is true for
  $i=0$. Let us assume that it is true for some $i$: $n_i \leq \ell n
  \left(1 - \frac{x}{\ell m}\right)^i$. Let $F_i$ denote the set of
  agents that have free slots after iteration $i$. There are at least
  $\frac{n_i}{\ell}$ such agents.  For $j \in F_i$, let $S(j)$ be the
  number of $j$'s top-$x$ items that were not included in the solution
  yet. If $j \in F_i$ has $s$ free slots, then $S(j) =
  (x-\ell+s)$. Thus we have that $\sum_{j \in F_i}S(j) \geq n_i + (x -
  \ell)\frac{n_i}{\ell} = \frac{n_ix}{\ell}$. By the pigeonhole
  principle, there exists an item that is among top-$x$ items for at
  least $\frac{n_ix}{\ell m}$ agents from $F_i$. Thus, after the
  $(i+1)$-th iteration of the algorithm, the total number of free
  slots is at most:
  \begin{align*}
    \textstyle
    n_{i+1} &\leq n_i - \frac{n_ix}{\ell m} = n_i\left(1 - \frac{x}{\ell m} \right) 
    \leq \ell n \left(1 - \frac{x}{\ell m}\right)^{(i+1)}.
  \end{align*}
  The number of free slots after the last iteration is at most:
  \begin{align*}
   \textstyle
    n_{K} \leq  \ell n \left(1 - \frac{x}{\ell m}\right)^{K}\!\!\! = \ell n \left(1 - \frac{\gamma}{\ell}\right)^{K}  \!\!\! \leq \ell n \exp\left(-\frac{\gamma K}{\ell}\right) \textrm{.}
  \end{align*}
  Thus the number of occupied slots is at least $\ell n - \ell n
  \exp(-\frac{\gamma K}{\ell})$. Every item that occupies an agent's
  slot has utility for this agent at least $\beta u_{\max}$, where
  $u_{\max}$ is the maximal utility that any of the agents assigns to
  an item.

  It remains to assess the OWA coefficients for the utilities of the
  items in the solution. If for some agent $i$ the utility of an item
  $a$, $u_{i, a}$, is taken with coefficient $\alpha_p$ ($p > 1$),
  then in the solution there must be an item $a'$ such that $u_{i, a'}
  \geq u_{i, a}$ and $u_{i, a'}$ is taken with coefficient
  $\alpha_{p-1}$. So there must exist at least $\frac{1}{\ell}(\ell n
  - \ell n \exp(-\frac{\gamma K}{\ell}))$ occurrences of the items
  whose utilities are taken with coefficient $\alpha_1$. By
  repeating this reasoning for the remaining occurrences of the items
  from the solution, since $\alpha$ is
  nonincreasing, we get that the total utility of the agents is at
  least
    $\textstyle \beta u_{\max} (\ell n - \ell n \exp(-\frac{\gamma K}{\ell}))\frac{1}{\ell}\sum_{i = 1}^{\ell}\alpha_i =
    \textstyle \beta u_{\max} n (1 - \exp(-\frac{\gamma K}{\ell}))\sum_{i = 1}^{\ell}\alpha_i.$
   Since no solution has utility higher than $nu_{\max}\sum_{i
      = 1}^{\ell}\alpha_i$, we get our approximation ratio.
%
\end{proof}

As a consequence, we get very good approximation guarantees for the
case of Borda-based utilities. Recall that $\w(\cdot)$ is Lambert's
$\w$ function, that is, a function that for $x \in \realsplus$
satisfies the equation $x = \w(x)e^{\w(x)}$ (and, thus, $\w(x)$ is
$O(\log(x))$).

\begin{corollary}\label{cor:borda-ptas}
  Fix a positive integer $\ell$ and let $\owa$ be a nonincreasing OWA
  where at most first $\ell$ entries are nonzero. Assume that agents
  have Borda-based utilities.  With $x =
  m\w\left(\frac{K}{\ell}\right)\frac{\ell}{K}$,
  Algorithm~\ref{alg:greedy2} is a $\left(1 -
    \frac{2\w(K/\ell)}{K/\ell}\right)$-approximation algorithm for
  $\alpha$-\textsc{OWA-Winner}.
\end{corollary}
\begin{proof}
  Let us note that the Borda utilities are
  $\left(1-\frac{\w(K/\ell)}{K/\ell},
    \frac{\w(K/\ell)}{K/\ell}\right)$--non-finicky. By applying
  Theorem~\ref{thm:borda:nonincreasing}, we get the following
  approximation ratio (the last equality follows by the definition of
  $\w(x)$):
\begin{align*}
\mathrm{approx.\ ratio} &= \left(1-\frac{\w(K/\ell)}{K/\ell}\right)\left(1 - \exp\left(-\frac{\left(\frac{\w(K/\ell)}{K/\ell}\right) K}{\ell}\right)\right) \\
                &= \left(1-\frac{\w(K/\ell)}{K/\ell}\right)\left(1 - \exp\left(-\w(K/\ell)\right)\right) \\
                &= \left(1-\frac{\w(K/\ell)}{K/\ell}\right)\left(1-\frac{\w(K/\ell)}{K/\ell}\right) \geq \left(1-\frac{2\w(K/\ell)}{K/\ell}\right) \textrm{.}
\end{align*}
This completes the proof.
\end{proof}

The next corollary follows directly from
Theorem~\ref{thm:borda:nonincreasing} by noting that in the case of
$m$ items and $k$-approval utilities (i.e., for the case where each
agent approves of exactly $k$ items) we have
$(1,\frac{k}{m})$-non-finicky utilities.

\begin{corollary}
  Fix a positive integer $\ell$ and let $\owa$ be a nonincreasing OWA
  where at most first $\ell$ entries are nonzero. Assume the $k$-approval utilities of the agents. 
  Algorithm~\ref{alg:greedy2} is a $\left(1 - \exp\left(-\frac{k K}{\ell m}\right)\right)$-approximation
  algorithm for $\alpha$-\textsc{OWA-Winner}.
\end{corollary}

\begin{figure*}[th!]
\begin{center}
\begin{minipage}[h]{0.48\linewidth}
  \centering
  \includegraphics[width=\textwidth]{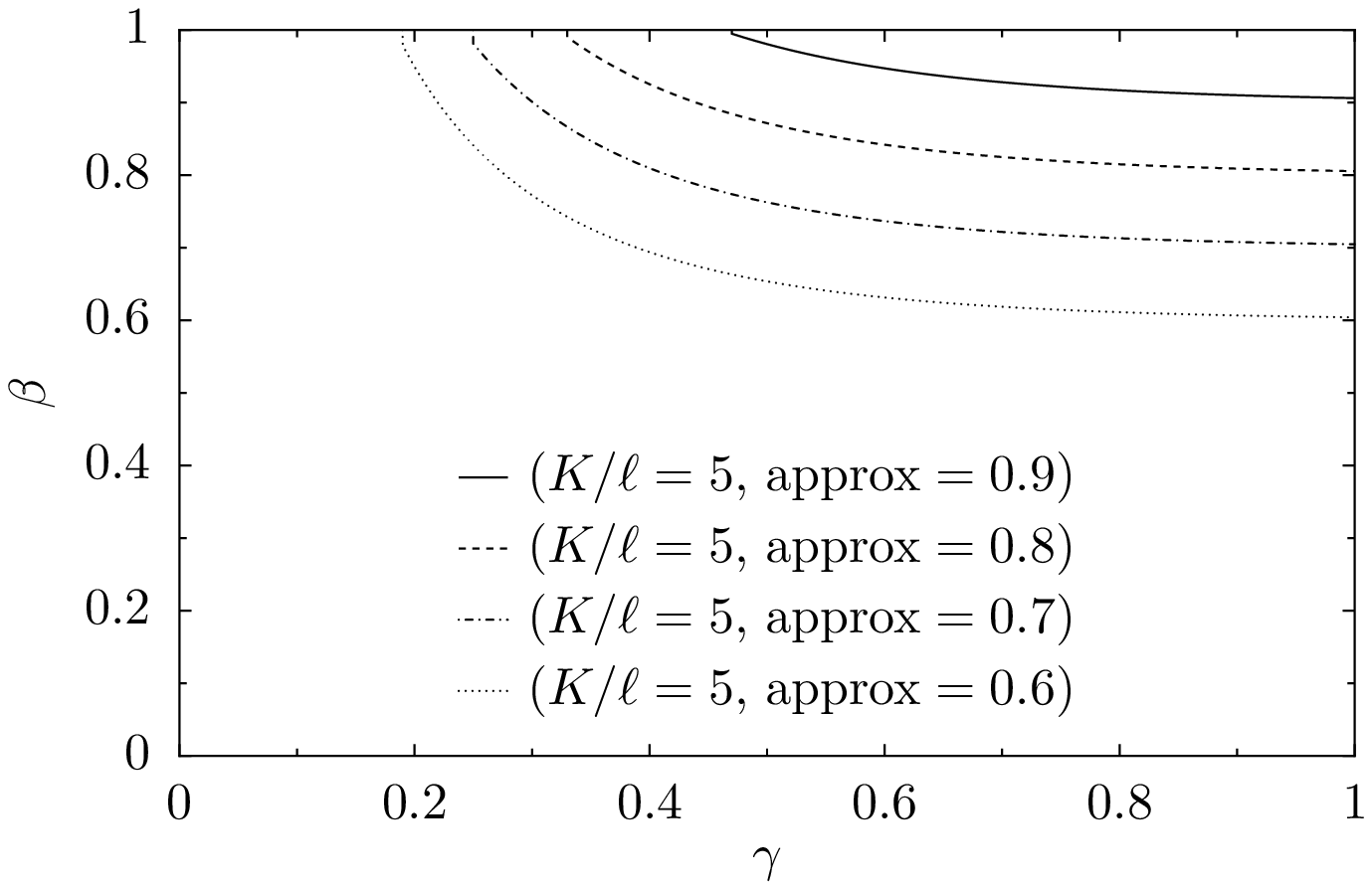}
  (a)
\end{minipage}
\begin{minipage}[h]{0.48\linewidth}
  \centering
  \includegraphics[width=\textwidth]{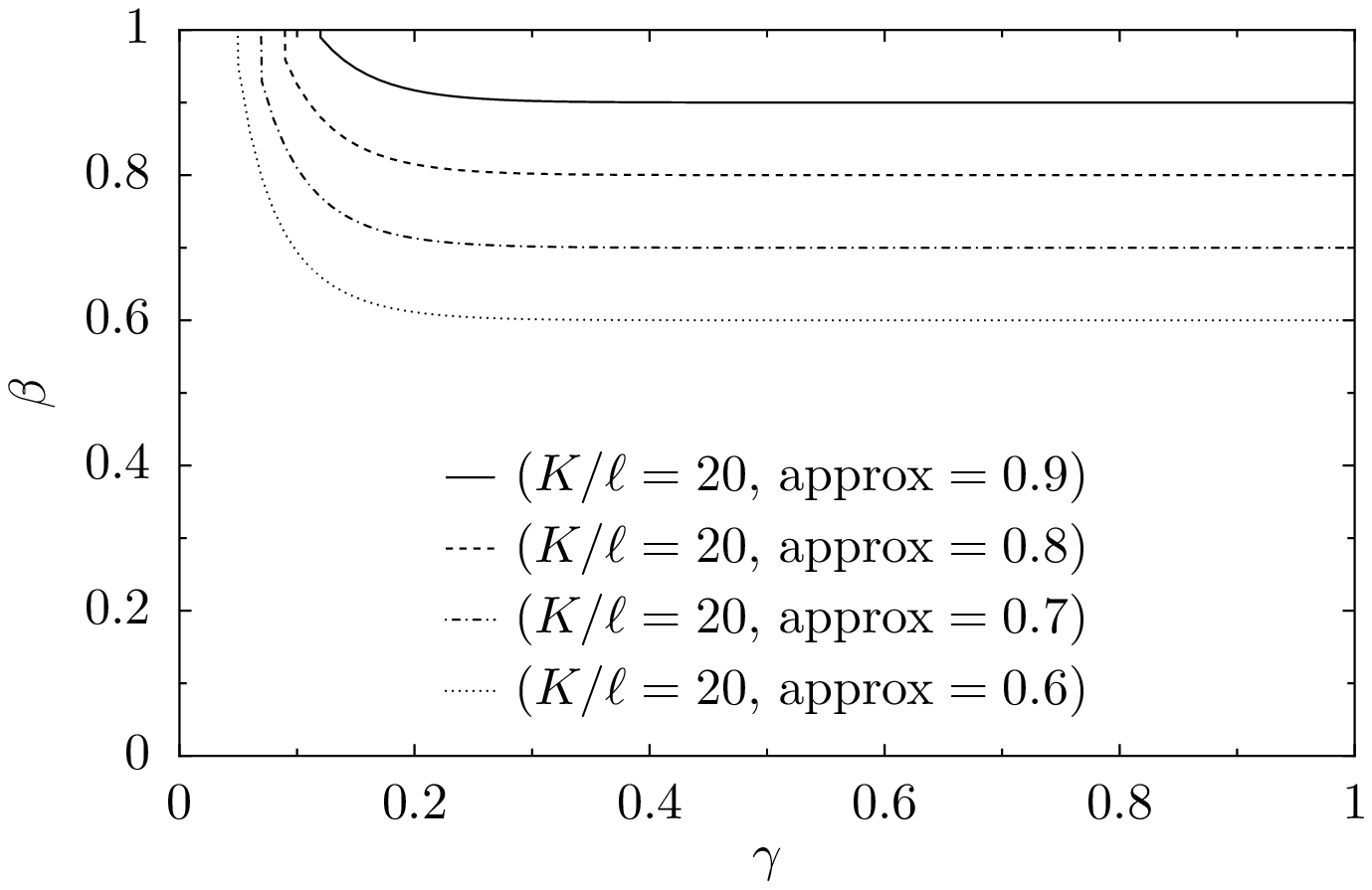}
  (b)
\end{minipage}
\begin{minipage}[h]{0.48\linewidth}
  \centering
  \includegraphics[width=\textwidth]{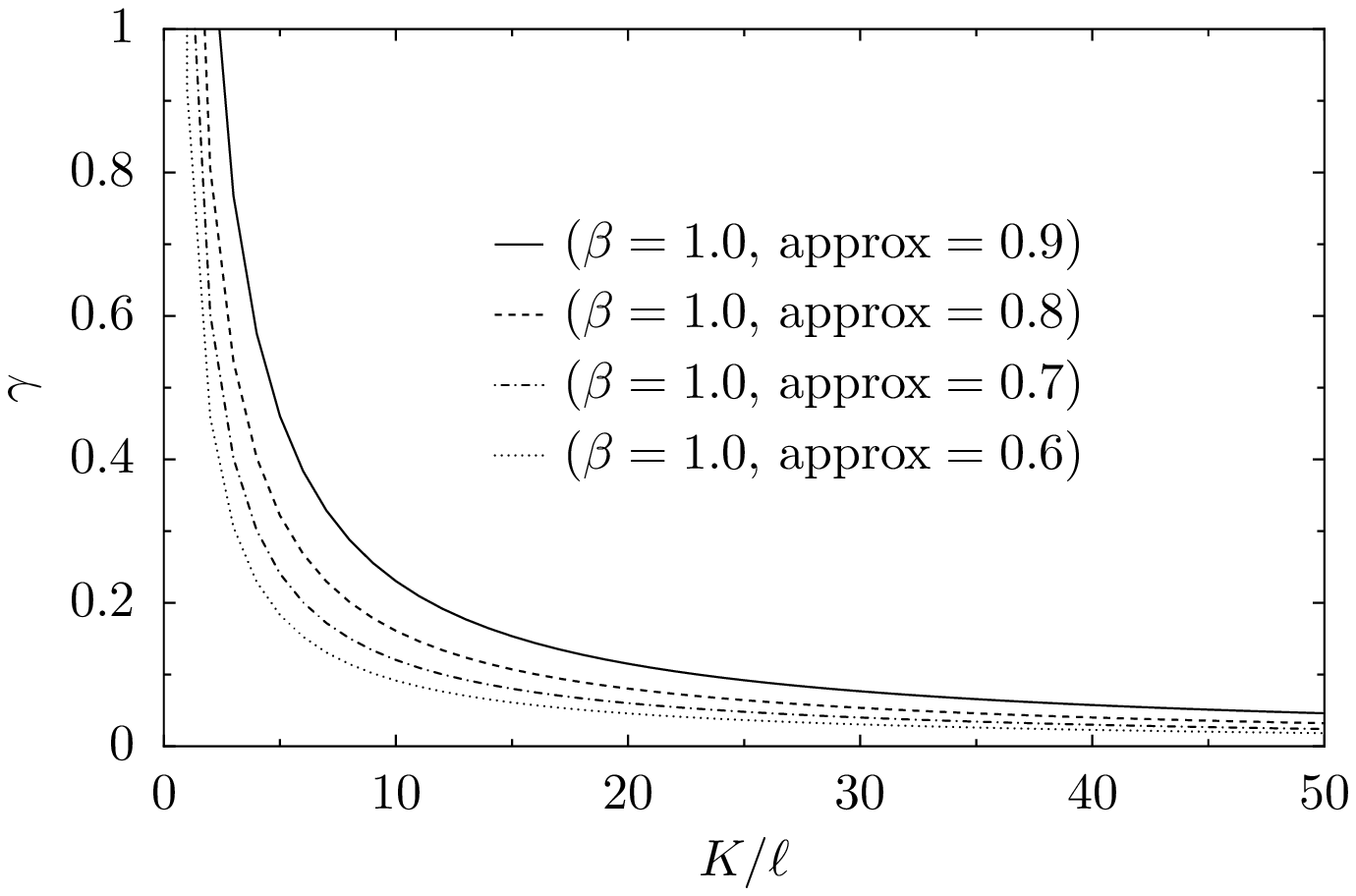}
  (c)
\end{minipage}
\begin{minipage}[h]{0.48\linewidth}
  \centering
  \includegraphics[width=\textwidth]{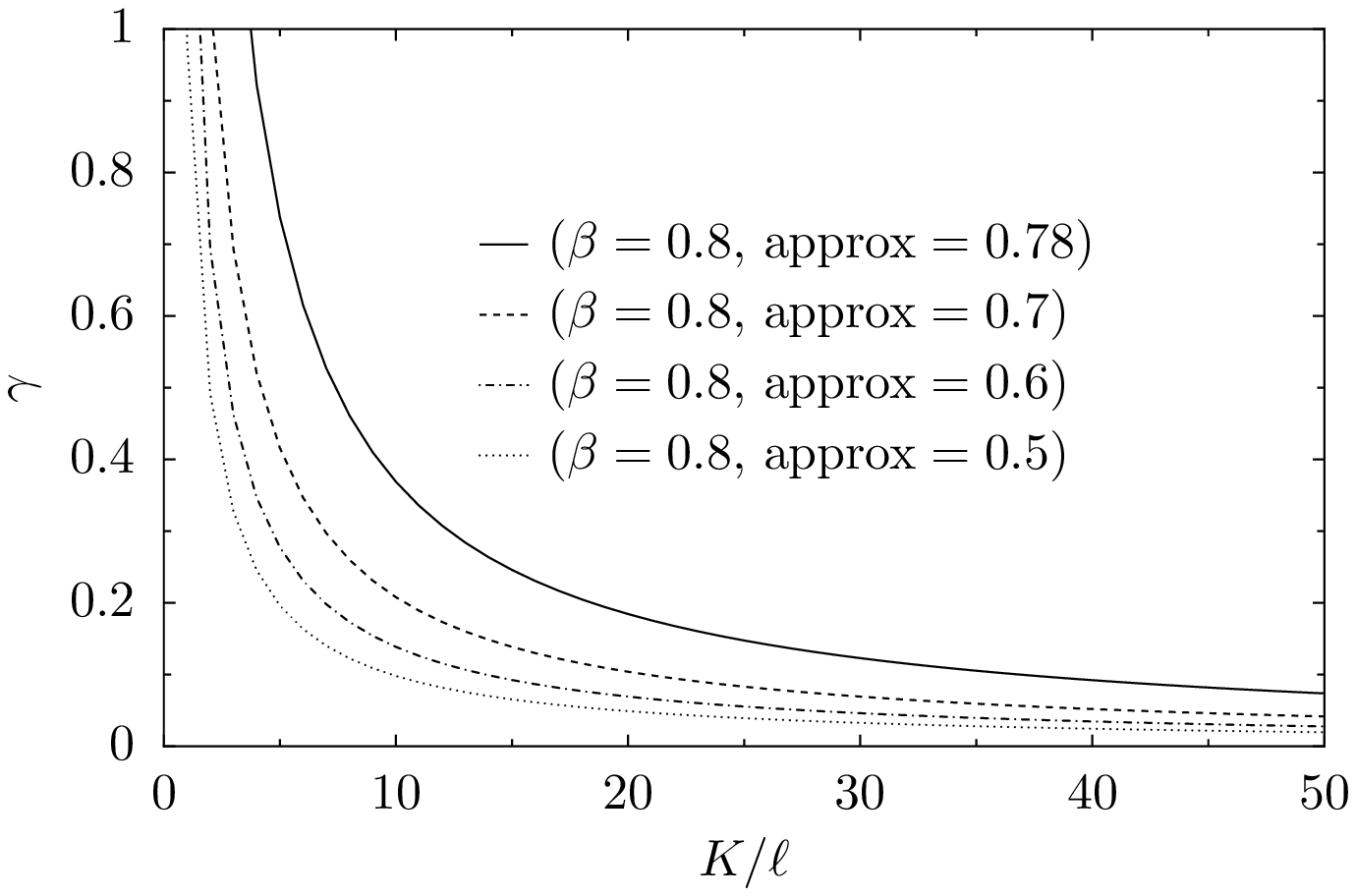}
  (d)
\end{minipage}
\end{center}
\vspace{-0.3cm}
\caption{The approximation ratios of Algorithm~\ref{alg:greedy2} for a
  nonincreasing OWA with at most $\ell$ top coefficients greater than
  zero, for $(\beta,\gamma)$-non-finicky utilities. The lines in
  Figures (a) and (b) depict the relations between the parameters
  $\beta$ and $\gamma$ that, for a given fixed ratio $\frac{K}{\ell}$,
  lead to the same approximation bound. The lines in Figures (c) and (d) depict
  the relations between the parameter $\gamma$ and the ratio $K/\ell$
  that, for a given fixed value of the parameter $\beta$ lead to the
  same approximation bound.}
\label{fig:finicky_approx}

\end{figure*}

Approximation ratio of Algorithm~\ref{alg:greedy2} is particularly
good when $K$ is large compared to $\ell$. This, indeed, is the most
interesting case because for small $K$ we can find optimal solutions
by brute-force search (combining these two approaches leads to a PTAS;
see Theorem~\ref{thm:borda:any-first-ell-ptas} below).  Nevertheless,
Algorithm~\ref{alg:greedy2} often gives a satisfactory approximation
guarantees by itself.  Figure~\ref{fig:finicky_approx} depicts the
classes of non-finicky utilities for which, for a fixed ratio
$K/\ell$, Algorithm~\ref{alg:greedy2} guarantees appropriate
approximation ratios: Parts (a) and (b) of the figure show the
relation that $\beta$ and $\gamma$ have to satisfy to obtain a
particular approximation ratio, for a given value
$\frac{K}{\ell}$. Part (c) shows the relation between the value of
$\gamma$ and the ratio $\frac{K}{\ell}$ that has to be satisfied for
Algorithm~\ref{alg:greedy2} to achieve a particular approximation
ratio under $(1,\gamma)$-non-finicky utilities, and part (d) shows the
same relation for $(0.8,\gamma)$-non-finicky utilities.

Theorem~\ref{thm:borda:nonincreasing} can be generalized to the case
of OWAs that are not nonincreasing (achieving a slightly weaker
approximation ratio).

\begin{lemma}\label{lem:borda:pghl}
  Consider a set $N$ of $n$ agents and a set $A$ of $m$ items,
  where the agents rank the items from the most preferred ones
  to the least preferred ones. Let $K$, $p$, and $t$ be some
  positive integers such that $K \leq m$, $p \leq K$, and $t \leq
  p$.  Let $x = \frac{\gamma}{p} m$. There is a
  polynomial-time algorithm that finds a collection $C$ of up to
  $K/p$ items such that there are at least
  $n\left(1-\exp\left(-\frac{\gamma K}{p^2}\right)\right)$ agents that each rank at least one
  member of $C$ between positions $(t-1)x+1$ and $tx$.
\end{lemma}
\begin{proof}
  To see that this lemma holds, it suffices to analyze the proof of
  Theorem~\ref{thm:borda:nonincreasing} for $1$-best-OWA, with
  $(1,\frac{\gamma}{p})$-non-finicky utilities, seeking winner set of
  size $\frac{K}{p}$. We note that the proof works equally well
  irrespectively of whether we consider the positions $1$ through $x$,
  or $x+1$ through $2x$, or any other segment of $x$ positions in the
  agents' preference orders.
\end{proof}

\begin{theorem}\label{thm:borda:any-first-ell}
  Fix a positive integer $\ell$ and let $\alpha$ be a family of OWAs
  that have nonzero entries on top $\ell$ positions only.  There is a
  polynomial-time $\beta\left(1 - \ell \exp\left(-\frac{\gamma
        K}{\ell^2}\right)\right)$--approximation algorithm for
  $\alpha$-\textsc{OWA-Winner} for the case of $(\beta,
  \gamma)$--non-finicky utilities.
\end{theorem}
\begin{proof}
  Consider an input instance $I$ of $\alpha$-\textsc{OWA-Winner} with
  the set $N = [n]$ of agents, with the set $A$ of $m$ items, and
  where we seek winner set of size $K$. Let $\alpha = \langle\alpha_1,
  \ldots, \alpha_\ell,0,\ldots,0\rangle$ be the OWA used in this
  instance. We set $x = \frac{\gamma}{\ell} m$.

  Our algorithm proceeds in $\ell$ iterations. We set $N^{(0)} = N$
  and $n^{(0)} = n$. In the $i$-th iteration, $1 \leq i \leq \ell$,
  the algorithm operates as follows: Using the algorithm from
  Lemma~\ref{lem:borda:pghl}, for $p = \ell$, we find a set $A^{(i)}$
  of up to $K/\ell$ items such that at least
  $n^{(i-1)}\left(1-\exp\left(-\frac{\gamma K}{\ell^2}\right)\right)$
  of the agents from the set $N^{(i-1)}$ each rank at least one of
  these items among positions $(i-1)x+1, \ldots, ix$ of their
  preference orders.  (Strictly speaking, in this setting agents do
  not have preference orders but utility values. For each agent, we
  form his or her preference order by sorting the items in the
  decreasing order of the utlities, breaking the ties arbitrarily.)
  We let $N^{(i)}$ be the set of these agents and we set $n(i) =
  \|N^{(i)}\|$. Finally, we set $W = \bigcup_{i=1}^{\ell}A^{(i)}$ and
  return $W$ as the set of winners (it is easy to see that $W$
  contains at most $K$ items; if $K$ contains fewer than $K$ items
  then we supplement it with $K-\|W\|$ arbitrarily chosen ones).

  By the construction of our algorithm, each of the agents from the
  set $N^{(\ell)}$ ranks at least $\ell$ items from the set $W$ on
  positions no worse than $\ell x = \gamma m$. Thus each such an agent
  assigns to each such an item utility at least equal to $\beta
  u_{\max}$. Consequently, the total utility that the agents from the
  set $N$ derive from the solution $W$ is at least:
  \[
     n^{(\ell)}\left(\sum_{i=1}^{\ell}\alpha_i\right)\beta u_{\max}.
  \]
  This is so, because for each $i$, $1 \leq i \leq \ell$, each of the
  agents in the set $N^{(\ell)}$ derives utility $\alpha_i \beta u_{\max}$
  from the item that she ranks as $i$'th best among the items
  from $W$.

  By construction of our algorithm, we have:
  \[n^{(\ell)} \geq n\left(1-\exp\left(-\frac{\gamma K}{\ell^2}\right)\right)^\ell \geq n\left(1- \ell \exp\left(-\frac{\gamma K}{\ell^2}\right)\right).\]
  Thus, the total utility obtained by the agents is at least:
  \begin{align*}
    u_{ut}^{\owa}(W) \geq n\left(1- \ell \exp\left(-\frac{\gamma
          K}{\ell^2}\right)\right)\left(\sum_{i=1}^{\ell}\alpha_i\right)\beta
    u_{\max} \textrm{.}
  \end{align*}
  Now, since the maximum possible total utility of all the agents is
  upper-bounded by $n(\sum_{i=1}^{\ell}\alpha_i) u_{\max}$, we have
  that our algorithm has approximation ratio $\beta\left(1- \ell \exp\left(-\frac{\gamma K}{\ell^2}\right)\right)$. It is clear that it runs in polynomial time, and so the
  proof is complete.
\end{proof}

Based on this result, we can obtain a PTAS for the analogous setting.

\begin{theorem}\label{thm:borda:any-first-ell-ptas}
  Fix a value $\ell$ and let $\alpha$ be a family of OWAs that have
  nonzero values on top $\ell$ positions only. There is a PTAS for
  $\alpha$-\textsc{OWA-Winner} for the case of (i) Borda-based
  utilities, and (ii) $(1, \gamma)$--non-finicky utilities (assuming
  $\gamma$ is a constant).
\end{theorem}
\begin{proof}
  For every $\epsilon$ we show a polynomial-time algorithm with
  approximation ratio $(1-\epsilon)$. Consider some $\epsilon$, $0
  \leq \epsilon \leq 1$. There exists a value $K_\epsilon$ such that
  for each $K > K_\epsilon$ it holds that $\ell
  \exp\left(-\frac{\gamma K}{\ell^2}\right) < \epsilon$. For each
  instance $I$ of $\alpha$-\textsc{OWA-Winner} where we seek winner
  set of size at least $K_\epsilon$, we run the algorithm from
  Theorem~\ref{thm:borda:any-first-ell}. For the remaining cases,
  where the winner-set size is bounded by a constant, we use a
  brute-force algorithm.
\end{proof}

We can also obtain a PTAS for \textsc{OWA-Winner} for geometric
progression OWAs for these classes of utilities. In essence, for
geometric progression it suffices to focus on a small number of top
entries in the OWA vector.  This is quite a useful result: Some of our
scenarios from Section~\ref{sec:scenarios} yield OWAs of this form.

\begin{corollary}\label{cor:borda-geometric-ptas}
  Fix a value $p > 1$. There is a PTAS for
  $\mathrm{gprog}[p]$-\textsc{OWA-Winner} for the case of (i) Borda-based
  utilities, and (ii) $(1, \gamma)$--non-finicky utilities (assuming $\gamma$ is a constant).
\end{corollary}
\begin{proof}
  Our goal is to show an algorithm that for a given value $\epsilon$,
  $\epsilon > 0$, in polynomial time outputs a
  $(1-\epsilon)$-approximate solution for
  $\mathrm{gprog}[p]$-\textsc{OWA-Winner}. Let us fix the value of
  such $\epsilon$. The idea of our proof is to truncate the vector
  describing $\mathrm{gprog}[p]$ OWA to consider only some $\ell$
  nonzero items on the top, where $\ell$ depends on $\epsilon$ only,
  and to run the algorithm from
  Theorem~\ref{thm:borda:any-first-ell-ptas}.

  For a given number $t$, let $S_{t}$ be the sum of the first $t$ coefficients
  of $\mathrm{gprog}[p]$. We have:
  \begin{align*}
    S_{t} &= \mathrm{gprog}[p]_t + \mathrm{gprog}[p]_{t-1} + \dots + \mathrm{gprog}[p]_1 \\
          &= p^{K-t} + p^{K-(t-1)} + \dots + p^{K-1} = p^{K-t}\frac{p^{t}-1}{p-1}.
  \end{align*}
  We fix $\ell = \lceil \log_{p}(\frac{2}{\epsilon}) \rceil$. Now,
  consider the ratio $r = S_{\ell}/S_K$:
  \begin{align*}
    r = \frac{S_{\ell}}{S_K} = p^{K-\ell}\frac{p^{\ell}-1}{p^K-1}
       > p^{K-\ell}\frac{p^{\ell}-1}{p^K}
      = 1 - \frac{1}{p^{\ell}} \geq 1 -  \frac{1}{p^{\log_{p}(\frac{2}{\epsilon})}} = 1 - \frac{\epsilon}{2} \text{.}
  \end{align*}
  Intuitively, the above inequality says that $1 - \frac{\epsilon}{2}$
  fraction of the total weight of $\mathrm{gprog}[p]$ OWA is
  concentrated in its first $\ell$ coefficients.

  Let $\mathrm{gprog}[p]_{|\ell}$ denote the OWA obtained from
  $\mathrm{gprog}[p]$ by replacing all coefficients with indices
  greater than $\ell$ with 0. Let $\mathcal{A}$ be a $(1 -
  \frac{\epsilon}{2})$-approximation algorithm for
  $\mathrm{gprog}[p]_{|\ell}$-\textsc{OWA-Winner}. From
  Theorem~\ref{thm:borda:any-first-ell-ptas} we know that such an
  algorithm exists. It is easy to see that $\mathcal{A}$ is a $(1 -
  \epsilon)$-approximation algorithm for
  $\mathrm{gprog}[p]$-\textsc{OWA-Winner}. Indeed, the utility under
  $\mathrm{gprog}[p]_{|\ell}$ for every $K$-element set $W$ is close
  to the utility of $W$ under $\mathrm{gprog}[p]$ (recall the $\vec
  x^{\downarrow}$ notation for sorted sequences; the inequality in the
  second line holds because for each $i$ we have $\sum_{g=1}^\ell
  \mathrm{gprog}[p]_gu^{\downarrow}_{i,w_h} \leq \sum_{j=1}^\ell
  \mathrm{gprog}[p]_ju^{\downarrow}_{i,w_j}$):
{\footnotesize
  \begin{align*}
    u_{ut}^{\mathrm{gprog}[p]}(W) &= \sum_{i=1}^n\sum_{j=1}^K \mathrm{gprog}[p]_j u^{\downarrow}_{i, w_j}
     \leq \sum_{i=1}^n\Big(\sum_{j=1}^\ell \mathrm{gprog}[p]_j u^{\downarrow}_{i, w_j} +  \sum_{h=\ell + 1}^K \mathrm{gprog}[p]_h \frac{\sum_{j=1}^\ell \mathrm{gprog}[p]_j u^{\downarrow}_{i, w_j}}{\sum_{g=1}^\ell \mathrm{gprog}[p]_g} \Big)\\
    & = \sum_{i=1}^n\sum_{j=1}^\ell \mathrm{gprog}[p]_j u^{\downarrow}_{i, w_j}\left(1 + \frac{\sum_{h=\ell + 1}^K \mathrm{gprog}[p]_h}{\sum_{g=1}^\ell \mathrm{gprog}[p]_g} \right)
    \leq \sum_{i=1}^n\sum_{j=1}^\ell \mathrm{gprog}[p]_j u^{\downarrow}_{i, w_j}\left(1 + \frac{\epsilon}{2}\right) \\
    & = \left(1 + \frac{\epsilon}{2}\right) u_{ut}^{\mathrm{gprog}[p]_{|\ell}}(W)
    \text{.}
  \end{align*}
}
  From which we get that for every $W$:
  \begin{align*}
    u_{ut}^{\mathrm{gprog}[p]_{|\ell}}(W) \geq (1 -
    \frac{\epsilon}{2})u_{ut}^{\mathrm{gprog}[p]}(W) \text{.}
  \end{align*}
  This completes the proof because algorithm $\calA$ returns a
  $(1-\frac{\epsilon}{2})$-approximate solution for
  $\mathrm{gprog}[p]_{|\ell}$-\textsc{OWA-Winner} and
  $(1-\frac{\epsilon}{2})(1-\frac{\epsilon}{2}) \geq 1-\epsilon$.
\end{proof}

To summarize, in this section we have shown that in spite of the
intrinsic hardness of the \textsc{OWA-Winner} problem, there are very
natural classes of utilities and OWA vectors for which the problem can
be solved quite accurately and very efficiently.

\section{Related Work}\label{sec:related}

In this section we give a more detailed overview of various research
lines that are related to our work.

Weighing intrinsic values by coefficients that are a function of their
rank in a list is of course not new. Ordered Weighted Average
operators have been used extensively in multicriteria decision making
and, to a lesser extent, in social
choice~\cite{DBLP:series/sfsc/KacprzykNZ11}; the vector of values then
corresponds to criteria (in MCDM) or to agents (in social
choice). Also, {\em rank-dependent expected utility} (RDEU)
\cite{Quiggin93} is a well-known research stream in decision theory,
whose starting point is the construction of models that explain
Allais' paradox: given a set of possible consequences of an act, the
contribution of a possible consequence on the agent's RDEU is a
function of its probability and of its rank in the list of
consequences ordered by decreasing probability. While these three
research streams use ranks to modify the contribution of a {\em
  criterion}, an {\em agent}, or a {\em possible consequence}, in our
setting they modify the contribution of {\em items}, our final aim
being to select an optimal set of items. Since we do not select
criteria, agents or possible consequences, it is not obvious how our
results can apply to these three aforementioned research
fields.

There are three recent pieces of research that use OWA operators in
the context of voting and that call for detailed discussion. We
describe them in the chronological order.

Goldsmith et al.~\cite{GLMP14} define {\em rank-dependent} scoring
rules.  Under standard positional scoring rules, the score of a
candidate is the sum of the scores it obtains from all the voters,
where the score that a candidate obtains from a given voter depends
only on his or her rank in this voter's preference
order. Rank-dependent scoring rules generalize this idea as
follows. Instead of simply summing up the scores of a given candidate,
they apply an OWA operator to the list of the scores that he or she
got from the voters.
Thus a rank-dependent scoring rule is
defined by a scoring vector (a function mapping ranks to scores) and
an OWA operator. Here, OWAs replace the sum operator for aggregating
the scores coming from {\em different agents}, while in our setting
they aggregate the scores of {\em different object} for a fixed agent.

Amanatidis et al.~\cite{ABLMR15} define a family of committee election
rules (which can also be used for multiple referenda) based on the
following principle. Each voter specifies his or her preferred
committee and each voter's disutility for a committee is given by the
Hamming distance between the committee and the voter's preferred
one. Then the disutilities of the voters are aggregated using an OWA
operator. The committee with the lowest aggregated disutility wins.
(In the particular case of the sum operator, the obtained rule is the
Bloc committee election rule, while in the case of the minimum, the
obtained rule is the {\em Minimax Approval Voting} rule; see the work
of Brams et al.~\cite{mavVoting} for the definition and other works
for computational
discussions~\cite{leg-mar-meh:c:approval-minimax,byr-sor:c:minimax-approval-ptas,mis-nab-sin:c:minimax-approva}.)
They obtain a number of hardness and approximability results, which
cannot be compared to ours because in their work, again, OWAs are used
for aggregating scores coming from {\em different agents}.

Finally, the work of Elkind and Isma\"ili~\cite{elk-ism:owa-cc} is
probably the closest one to ours. They study multiwinner elections and
they use OWAs to define generalizations of the Chamberlin--Courant
rule but, once again, they use OWAs to aggregate the utilities for a
committee coming from different agents. The standard utilitarian
Chamberlin--Courant rule sums up the scores that a committee gets from
different voters, whereas the egalitarian variant considers the
minimum score a committee receives. They generalize this idea by using
an OWA operator, in effect obtaining a spectrum of rules between the
utilitarian and the egalitarian variants. They obtain a number of
complexity results, both in the general case and in specific cases
corresponding to domain restrictions. For the same reason as in the
preceding paragraphs, their results are incomparable to ours.

In the three pieces of research discussed above, OWA operators
aggregate scores or utilities given to candidates or committees by
different agents, which is very different from our use of
OWAs. Nonetheless, there exists a high-level common point between the
four approaches.  In all cases the rules corresponding to the sum of
scores, and to either the minimum or the maximum of scores, were
already known and seen as interesting, but somewhat extreme. In all
cases, OWAs give rise to an interpolation between these extremities,
leading to rules and approaches that are likely to be interesting in
practice.


Let us now move on to other related works and other related streams of
research.  Several known settings are recovered as particular cases of
our general model. In particular, this applies to the case of the
Chamberlin--Courant proportional representation
rule~\cite{ccElection}, to the case of Proportional Approval
Voting~\shortcite{pavVoting}, and to (variants of) the budgeted social
choice
model~\cite{budgetSocialChoice,ore-luc:c:cc-online,bou-lu:c:value-directed-cc}.
Computational complexity of the Chamberlin--Courant rule was first
studied by Procaccia et al.~\shortcite{complexityProportionalRepr},
its parameterized complexity was analyzed by
Betzler~et~al.~\shortcite{fullyProportionalRepr}, and the complexity
under restricted domains was studied by
Betzler~et~al.~\shortcite{fullyProportionalRepr}, Skowron et
al.~\cite{sko-yu-fal-elk:j:mwsc}, Yu et
al.~\cite{yu-cha-elk:c:sp-tree}, and Clearwater et
al.~\cite{cle-pup-sli:c:median-graph-sc-cc}. The first approximation
algorithm was proposed by
Lu~and~Boutilier~\shortcite{budgetSocialChoice}. The results on
approximability were then extended in several directions by
Skowron~et~al.~\shortcite{sko-fal-sli:c:multiwinner,sko-fal:t:maxcover}. 
Proportional Approval Voting was studied computationally and
axiomatically by Aziz et
al.~\cite{azi-gas-gud-mac-mat-wal:c:approval-multiwinner,azi-bri-con-elk-fre-wal:c:justified-representation}
and by Elkind and Lackner~\cite{elk-lac:c:dichotomous-prefs}.

Group recommender systems (see, e.g., the work of O'Connor et
al.~\shortcite{OConnorCKR01} for one of the first approaches, and the
surveys of Jameson and Smyth~\shortcite{JamesonS07} and of
Masthoff~\shortcite{Masthoff10}) aim at recommending sets or sequences
of items (such as a set of television programs or a sequence of songs)
to a group of users, based on preferences of all group members. Two
mainstream approaches have been developed (see the survey of Jameson
and Smyth~\cite{JamesonS07}): those based on the construction of an
`average user' whose preferences are built by aggregating the
preferences of the individuals in the group, and those based on
producing individual recommendations and aggregating them. Unlike
these, our approach (which recommends sets, but not yet sequences)
proceeds in a single step, and enables a fine-tuning of the
contribution of an item to each user's utility depending on the number
of better items (for that user) in the list. 

The facility location problem (\textsc{fl}) is closely related to
$1\text{-}\mathrm{best}$-\textsc{OWA-Winner}. In \textsc{fl}, however,
the goal is to minimize the dissatisfaction of the agents instead of
maximizing their utility (satisfaction). Although, as far as exact
solutions are concerned both formulations are equivalent, there is a
significant difference in the quality of approximation (the difference
between approximation guarantees for the maximization and minimization
formulations of $1\text{-}\mathrm{best}$-\textsc{OWA-Winner} for Borda
utilities is described by
Skowron~et~al.~\shortcite{sko-fal-sli:c:multiwinner}). Some works
focus on general dissatisfaction
functions~\cite{Fellows:2011:FLP:1982694.1982895}, but most of the
results were established for dissatisfactions corresponding to the
distances, and thus satisfying the triangle
inequality~\cite{Jain:2001:AAM:375827.375845, ShmoysTA97}. Also, in
\textsc{fl} the goal is to minimize the dissatisfaction of the
worst-off agent (the egalitarian view). The utilitarian version of the
problem is called
\textsc{k-median}~\cite{Jain:2001:AAM:375827.375845}. The
parameterized complexity of the problem was analyzed by Fellows and
Fernau~\shortcite{Fellows:2011:FLP:1982694.1982895}. The approximation
algorithms include those of Chukad and
Williamson~\cite{Chudak:2005:IAA:1047770.1047776}, those of Jain and
Vazirani~\cite{Jain:2001:AAM:375827.375845}, and those of Shmoys et
al.~\cite{ShmoysTA97}. Interestingly, a local-search algorithm (which,
to the best of our knowledge, is the best known approximation
algorithm for the capacitated version of \textsc{fl}~\cite{
  Chudak:2005:IAA:1047770.1047776}) is also a
$\frac{1}{2}$-approximation algorithm for maximizing nondecreasing
submodular functions~\cite{submodular}, and thus for
\textsc{OWA-Winner} with non-decreasing utility functions. We conclude
that it would be interesting to compare the algorithms for \textsc{fl}
and \textsc{k-median} with different algorithms for
\textsc{OWA-Winner} on real preference data (e.g., on the data from
PrefLib, collected by Mattei and
Walsh~\cite{conf/aldt/MatteiW13}).

\section{Summary}\label{sec:summary}

Our contribution is threefold. First, we have proposed a new model for
the selection of a collective sets of items. This model appears to be
very general, encompasses several known frameworks, and can be applied
to various domains such as committee elections, group recommendation,
and beyond. Second, we have investigated the computational feasibility
of the model, depending on the various assumptions that we can make
about the agents' utilities and the choice of the OWA vector. Table
\ref{tab:summary} in Section~\ref{sec:overview} gives a summary of our
results. We note that many of these results directly related to the
OWA families that appear in the settings from
Section~\ref{sec:scenarios} that were our motivating force.  Third, we
have defined non-finicky utilities that model settings where agents
are relatively ``easy to please.'' We believe that non-finicky
utilities may find applications far beyond our framework.

Some of our results look negative, while some others (especially in
the case of non-finicky utilities) are on the positive side. However,
the way the results should be interpreted depends on the application
domain. In political elections and other high-stake domains, it is
appealing to view an approximation algorithm as a new, full-fledged
voting rule, which may enjoy many desirable properties (on this point
see the works of Caragiannis et
al.~\cite{car-cov-fel-hom-kak-kar-pro-ros:j:dodgson,car-kak-kar-pro:j:dodgson-acceptable},
Skowron et al.~\cite{sko-fal-sli:c:multiwinner}, and Elkind et
al.~\cite{elk-fal-sko-sli:c:multiwinner-rules}). In particular, we
have shown that the election system Sequential Proportional Approval
Voting, SPAV, (which has been known long before the computational
complexity theory was developed) is actually a greedy approximation
algorithm for the Proportional Approval Voting (PAV) election rule,
which is an interesting result {\em per se}.  (The reader may also
wish to consult the paper of Aziz et
al.~\cite{azi-gas-gud-mac-mat-wal:c:approval-multiwinner} regarding
the complexity of approval-based multiwinner rules.)  Yet, it is
arguably not reasonable to use an approximation algorithm (even with a
good performance guarantee) if it is viewed as nothing more than an
approximation algorithm of another rule, and it is even less
reasonable to use a heuristic search algorithm (when there is no good
approximation algorithm); this implies that using this model for
political elections is feasible when the number of candidates is small
enough, but can become problematic beyond that (unless we define the
approximation algorithm to be the new voting rule, as said above).  On
the other hand, in low-stake application domains (which can include
some committee elections, and of course group recommender systems), it
may become perfectly reasonable, and in that case even NP-hardness and
inapproximability results should not discourage us from using the
model. For these domains, our negative results only tell us that we
may have to resort to heuristic search algorithms. Developing such
algorithm is one of the interesting directions for further research.

Our work leads to many other open problems. In particular, one might
want to strengthen our approximation algorithms, provide algorithms
for more general cases, provide more inapproximability results. Among
these problems, a particularly interesting one regards the
approximability of \textsc{OWA-Winner} for the arithmetic progression
family of OWAs. For this case, our set of results is very limited.  In
particular, can one provide a PTAS for arithmetic-progression OWAs
under non-finicky (in particular Borda-based) utilities? Can one do so
for $\frac{K}{2}$-best OWAs/$K$-median OWAs?  Can one do so for the
harmonic OWA, used in Proportional Approval Voting?

\bibliographystyle{abbrv}
\bibliography{main}

\appendix

\section{Proofs of Theorem~\ref{thm:borda:k-1-best}}

\newtheorem*{tthmbordakminusonebest}{Theorem~\ref{thm:borda:k-1-best}}

\begin{tthmbordakminusonebest}
  \thmbordakminusonebest
\end{tthmbordakminusonebest}
\begin{proof}
  As before, it is clear that the problem is in an $\np$ and we only
  show $\np$-hardness. We give a reduction from \textsc{VertexCover}
  (see the previous proof for an exact definition).  Let $I$ be an
  instance of the \textsc{VertexCover} problem that consists of
  undirected graph $G = (V,E)$, where $V = \{v_1, \dots, v_m\}$ and $E
  = \{e_1, \dots, e_n\}$, and positive integer $K$ (without loss of
  generality, we assume that $K \geq 3$).

  From $I$, we construct an instance $I'$ of
  $(K-1)$-best-\textsc{OWA-Winner} with Borda-based utilities as
  follows.  We set \[x = 4n(m + 2)(K+4)\] and we let the set of
  items be $A = V \cup \{d_1, d_2\} \cup H$, where $H = \{h_1,
  \ldots, h_x\}$ and $\{d_1, d_2\}$ are sets of dummy items that we need to build
  appropriate structure of the utility profile.
  To build the set of agents $N$, we set 
  \[y = (n(x+m+2)^2 + 1)\] and we set $N = N_E \cup N_1 \cup \ldots
  \cup N_y$, where $N_E = \{e_1^1, e_1^2, \dots, e_n^1, e_n^2\}$
  contains pairs of agents that correspond to the edges of $G$, and $N_1,
  \ldots, N_y$ contain pairs of agents needed for the
  construction. Specifically, every set $N_i$, $1 \leq i \leq y$,
  consists of two agents, $f_{i}^{1}$ and $f_{i}^{2}$. We refer to the
  agents in the set $N_1 \cup \cdots \cup N_y$ as the ``dummy
  agents.''

  We describe agents' utilities through their preference orders.
  The agents in the set $N_E$ have the following preference 
  orders. Let $e_i \in E$ be an edge of the graph that connects
  vertices $v_{i, 1}$ and $v_{i, 2}$. Agents $e_i^1$ and $e_i^2$ have
  preference orders:
  \begin{align*}
    e_i^1&:  d_1 \succ d_2 \succ V - \{v_{i,1},v_{i,2}\} \succ H \succ \{v_{i,1},v_{i,2}\},\\
    e_i^2&:  d_1 \succ d_2 \succ \{v_{i,1},v_{i,2}\} \succ H \succ V - \{v_{i,1},v_{i,2}\}.
  \end{align*}
  (When we put a set of items in a preference order, this means
  that this set can be replaced by these items in an arbitrary,
  easily computable, way.)  Each agent $f_{i}^1$, $1 \leq i \leq y$,
  has the same, fixed, preference order:
  \begin{align*}
    f_{i}^{1}: d_1 \succ v_1 \succ v_2 \dots \succ v_m \succ d_2 \succ
    h_1 \dots \succ h_x \textrm{.}
  \end{align*}
  Similarly, each agent $f_{i}^2$, $1 \leq i \leq y$, has preference
  order:
  \begin{align*}
    f_{i}^{2}: d_2 \succ v_m \succ v_{m-1} \dots \succ v_1 \succ d_1
    \succ h_1 \dots \succ h_x \textrm{.}
  \end{align*}

  Finally, in the instance $I'$ we seek a set of winners of size
  $K+2$. This means that we use $(K+1)$-best-OWA to compute the
  aggregated utility than an agent derives from a set of winners.

  This concludes the description of the reduction and it is clear that
  it is polynomial-time computable. Before we prove that it is
  correct, let us make several observations.
  Let $W$ be some optimal solution for $I'$. We claim that $W$ does
  not contain any of the items from $H$.  For the sake of
  contradiction, assume that some $h \in H$ belongs to $W$. Since
  $d_1$ and $d_2$ are ranked ahead of $h$ in every preference order
  (and in some preference orders $d_1$ is first and $d_2$ is second,
  so their utility cannot be ignored by the $(K+1)$-best-OWA), we
  infer that $d_1$ and $d_2$ must belong to $W$ as well (otherwise we
  would obtain higher utility by replacing $h$ with one of $d_1$ and
  $d_2$ in $W$). Let $v$ be some item from $V$ that does not
  belong to $W$. If we replace $h$ with $v$ in $W$ then the total
  utility of the dummy agents increases by at least $2y$.  Why is this
  so? Consider some pair $N_i$, $1 \leq i \leq y$ of dummy
  agents. item $h$ is either the lowest ranked member of $W$
  for both $f_i^1$ and $f_i^2$ or for neither. We consider these cases:
  \begin{itemize}
  \item
  \textbf{$\boldsymbol{h}$ is the lowest-ranked winner for both the
    agents in $\boldsymbol{N_i}$.} Replacing $h$ with $v$ means that
    either some other member $h'$ of $H \cap W$ becomes the lowest
    ranked winner for both $f_i^1$ and $f_i^2$, or $d_2$ becomes the
    lowest ranked winner for $f_i^1$ and $d_1$ becomes the lowest
    ranked winner for $f_i^2$. In either case, both $f_i^1$ and
    $f_i^2$ obtain utility higher by at least one from $v$ than from
    the item that became the new lowest-ranked winner. Thus,
    the total utility yielded by these two agents increases by at
    least two.
  \item
  \textbf{$\boldsymbol{h}$ is not the lowest-ranked winner for either
    agent in $\boldsymbol{N_i}$.} In this case, since both agents rank
    $v$ higher thank $h$ and replacing $h$ with $v$ does not change
    the lowest-ranked winner for either of the agents, their total
    utility also increases at least by two.
  \end{itemize}
  Since there are $y$ pairs of agents, the total utility increases by
  at least $2y$. Since the total utility of the agents from $N_{E}$ is
  lower than $2n(x+m+2)^2 < 2y$, we see that after the change the
  total utility of all the agents increases. Thus, we get a
  contradiction and we conclude that $W$ does not contain any of the
  agents from $H$.

  Next, we claim that both $d_1$ and $d_2$ belong to $W$. We give a
  detailed argument for $d_1$ only; the case of $d_2$ is analogous.
  For the sake of contradiction, assume that $d_1$ does not belong to
  $W$. Let $v_k$ be an item from $W$ such for each $v_j$, $j <
  k$, $v_j$ does not belong to $W$.  By our assumptions, for each
  agent $f_i^2$, $1 \leq i \leq y$, $v_k$ is the lowest-ranked winner
  from $W$. Thus, if we replace $v_k$ with $d_1$ in $W$, then the
  utility of each agent $f_i^2$ will not change, whereas the utility
  of each agent $f_i^1$ will increase.  Further, the utility of each
  agent from $N_E$ will increase. Thus, by replacing $v_k$ with $d_1$,
  we can increase the total utility of the agents. We reach a
  contradiction and we conclude that $d_1$ must have been a member of
  $W$. An analogous argument shows that $d_2$ belongs to $W$ as well.

  As the result of the above reasoning, we infer that each set of
  winners consists of $d_1$, $d_2$, and $K$ items from $V$.
  Whenever both $d_1$ and $d_2$ are included in the set of winners and
  neither item from $H$ is, the total utility of the dummy
  agents is the same, irrespective which items from $V$ are
  selected. With these observations, we now show that the answer for
  the input \textsc{VertexCover} instance is ``yes'' if and only if
  there is a size-$(K+2)$ winner set for $I'$ that for agents in the
  set $N_E$ yields total utility at least $nx(K+4)$.

  $(\Rightarrow)$ Let us assume that there exists a cover $C$ for $I$,
  that is, a set $C$ of $K$ vertices such that each edge is incident
  to at least one vertex from $C$. We show that winner set $W = C \cup
  \{d_1,d_2\}$ gives total utility of every two agents $e_i^1$ and
  $e_i^2$, $1 \leq i \leq n$, equal to at least $x(K+4)$.  Pick some
  arbitrary $i$, $1 \leq i \leq n$, and let $v_{i,1}$ and $v_{i,2}$ be
  the two vertices connected by edge $e_i$.  If both $v_{i, 1}$ and
  $v_{i, 2}$ belong to $C$, then $e_i^2$ obtains utility at least $x$
  for each item in $\{v_{i,1},v_{i,2},d_1,d_2\}$ (at least
  utility $4x$ in total). On the other hand, $e_i^1$ obtains utility
  at least $x$ for each item in $W-\{v_{i,1},v_{i,2}\}$.  This
  gives utility at least $Kx$. Altogether, both agents get utility at
  least $x(K+4)$.  If only one of the items $v_{i, 1}$ and
  $v_{i, 2}$, say $v_{i, 1}$, belongs to $C$, then $e_i^2$ obtains utility at least $3x$
  (at least $x$ for every item from $\{v_{i, 1}, d_1, d_2\}$),
  and $e_i^1$ obtains utility at least $(K+1)x$ (at
  least $2x$ from items $d_1$ and $d_2$, and at least $(K-1)x$
  from the $K-1$ members of $C$ that $e_i^1$ ranks on the top
  positions). Again, both agents get utility at least $x(K+4)$.  Thus
  the total utility of the agents in $N_E$ in the optimal solution
  must be at least $nx(K+4)$.

  $(\Leftarrow)$ Assume that $W$ is some optimal solution for $I'$ and
  that for the agents in $N_E$ it yields utility at least
  $nx(K+4)$. By previous discussion, we know that $W$ contains $d_1$,
  $d_2$, and $K$ members of $V$. We set $C = W \setminus \{d_1,d_2\}$.
  Let us fix some arbitrary $i$, $1 \leq i \leq n$. Let $v_{i,1}$ and
  $v_{i,2}$ be the two vertices connected by edge $e_i$.  We observe
  that under $W$, the total utility of agents $e_i^1$ and $e_i^2$ is
  at most $(x + m + 2)(K+4) + mK$. To see this, let $z$ be the number
  of items from $\{v_{i,1},v_{i,2}\}$ that are included in $C$
  and note that (1) for the upper bound we can disregard the OWA that
  we use, (2) there are $x+m+2$ items and so we can upper-bound
  the utility derived from each item by $x+m+2$, (3)
  altogether, the items from $W$ are ranked on at most $K+2-z$
  top-$(m+2)$ positions by $e_i^1$ (we upper-bound their total utility
  by $(K+2-z)(x+m+2)$) and at most $2+z$ top-$(m+2)$ positions by
  $e_i^2$ (we upper-bound their total utility by $(2+z)(x+m+2)$), and
  (4) the items from $W$ are ranked on at most $z$ bottom-$m$
  positions by $e_i^1$ (we upper-bound their total utility by $zm$)
  and on $K-z$ bottom-$m$ positions by $e_i^2$ (we upper-bound their
  total utility by $(K-z)m$).  When we sum up these upper bounds, we
  get $(x + m + 2)(K+4) + mK$.  However, for our argument we also need
  an upper bound on the total utility of $e_i^1$ and $e_i^2$ under the
  assumption that neither $v_{i,1}$ nor $v_{i,2}$ belongs to $C$. In
  this case, the upper bound is $(x + m + 2)(K+3) + mK$. We obtain it
  in the same way as the previous bound, except that we note that due
  to our $(K+1)$-best-OWA, the utility derived by $e_i^1$ can take
  into account at most $K+1$ agents from the top-$(m+2)$ positions of
  the preference order of $e_i^1$.

  


  Based on these upper bounds, we will now show that if the total
  utility derived from $W$ by the agents in $N_E$ is $nx(K+4)$, then
  $C$ must correspond to a cover of all the edges of $G$. To this end,
  consider a situation where there is at least one edge $e_i$ such
  that neither of the vertices that it connects belongs to $C$. By
  using our upper bounds, in this case the total utility of the agents
  from $N_E$ can be at most:
  \begin{align*}
    &(K+3)(x+m+2) + (n-1)(K+4)(x+m+2) +nmK \\
    &= (x+m+2)( K+3 + (n-1)(K+4)) + nmK\\
    &= (x+m+2)( n(K+4) - 1) +nmK\\
    &= xn(K+4) + n(m+2)(K+4) -(x+m+2) +nmK\\
    &= xn(K+4) + 0.25x - (x+m+2) +nmK\\
    &< xn(K+4)
  \end{align*}
  (The last two lines follow directly by the definition of $x$.) So,
  from the assumption that $C$ is not a solution for $I$, we obtain
  that the total utility of the agents in $N_E$ must be lower than
  $nx(K+4)$, which contradicts our assumption. Thus $C$ is a correct
  solution for $I$ and, so, $I$ is a yes-instance of
  \textsc{VertexCover}. This completes the proof.
%
\end{proof}

\end{document}